\documentclass[10pt,twoside,letter]{scrartcl}

\usepackage[nonatbib]{styles/arxiv}

\usepackage[utf8]{inputenc} %
\usepackage[T1]{fontenc}    %
\usepackage[colorlinks]{hyperref}       %
\usepackage{url}            %
\usepackage{booktabs}       %
\usepackage{amsfonts}       %
\usepackage{nicefrac}       %
\usepackage{microtype}      %
\usepackage{multicol}
\usepackage{fancyhdr}
\usepackage{graphicx}
\usepackage{anyfontsize}
\usepackage{xcolor}
\usepackage{mathtools}

\usepackage{amsthm,thmtools}
\usepackage{styles/widetext}
\usepackage{bbm}
\usepackage{float}
\usepackage{cuted}
\usepackage{soul}
\usepackage{gensymb}

\usepackage{bookmark}
\usepackage{afterpage}

\usepackage{tikz-cd}
\usepackage{comment}
\usepackage{bm}

\hypersetup{
  citecolor=NavyBlue,
  linkcolor=NavyBlue,
  filecolor=NavyBlue,
  urlcolor=NavyBlue
}

\usepackage[
backend=biber,
style=bibstyle,
articletitle=true,
isbn=false,
natbib
]{biblatex}

\newcommand{\orcid}[1]{\href{https://orcid.org/#1}{\hspace{5pt}\includegraphics[width=10pt]{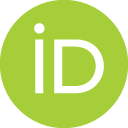}}}

\renewcommand{\thefootnote}{\fnsymbol{footnote}}

\usepackage{academicons}
\definecolor{orcidlogocol}{HTML}{A6CE39}

\addtokomafont{footnote}{\small\footnotefont}

\usepackage[english]{babel}
\newtheorem{theorem}{Theorem}

\newtheorem{corollary}{Corollary}

\newcommand{\refthm}[1]{\hyperref[#1]{Theorem}~\ref{#1}}
\newcommand{\reffig}[2]{\hyperref[#1]{Figure}~\ref{#1}\hyperref[#1]{#2}}
\renewcommand{\refeq}[1]{\hyperref[#1]{Equation}~\ref{#1}}

\newcommand{\beginsupplement}{%
        \setcounter{table}{0}
        \renewcommand{\thetable}{S\arabic{table}}%
        \setcounter{figure}{0}
        \renewcommand{\thefigure}{S\arabic{figure}}%
     }

\newcommand{\norm}[1]{\left\lVert #1 \right\rVert}

\DeclareMathOperator*{\argmax}{arg\,max}
\DeclareMathOperator*{\argmin}{arg\,min}

\definecolor{darkblue}{RGB}{8, 60, 125}
\definecolor{darkgreen}{RGB}{46, 151, 78}
\definecolor{darkpurple}{RGB}{97, 64, 155}
\definecolor{darkred}{RGB}{139, 0, 0}
\definecolor{navyblue}{HTML}{000DB2}

\newcommand{\kl}{\mathbb{D}_\text{KL}}

\date{\today} \title{Automated construction of cognitive maps \\  with visual predictive coding}

\author{%
  \href{mailto:jgornet@caltech.edu}{{James
    A. Gornet}}\textsuperscript{1,2}{\footnotemark}\orcid{0000-0002-5431-7340} \\
  {\small\href{mailto:jgornet@caltech.edu}{jgornet@caltech.edu}} \And
  \href{mailto:mthomson@caltech.edu}{{Matt Thomson}}\textsuperscript{1,2} \\
  {\small\href{mailto:mthomson@caltech.edu}{mthomson@caltech.edu}} }
\address{ \textsuperscript{1}California Institute of Technology,
  Division of Biology and Biological Engineering, Pasadena,
  CA, USA \\
  \textsuperscript{2}California Institute of Technology, Computation
  and Neural Systems, Pasadena, CA, USA }

\begin{document}
\maketitle

\begin{abstract}
  Humans construct internal cognitive maps of their environment directly from sensory inputs without access to a system of explicit coordinates or distance measurements. While machine learning algorithms like SLAM utilize specialized inference procedures to identify visual features and construct spatial maps from visual and odometry data, the general nature of cognitive maps in the brain suggests a unified mapping algorithmic strategy that can generalize to auditory, tactile, and linguistic inputs. Here, we demonstrate that predictive coding  provides a natural and versatile neural network algorithm for constructing spatial maps using sensory data. We introduce a framework in which an agent navigates a virtual environment while engaging in visual predictive coding using a self-attention-equipped convolutional neural network. While learning a next image prediction task, the agent automatically constructs an internal representation of the environment that quantitatively reflects spatial distances. The internal map enables the agent to pinpoint its location relative to landmarks using only visual information.The predictive coding network generates a vectorized encoding of the environment that supports vector navigation where individual latent space units delineate localized, overlapping neighborhoods in the environment. Broadly, our work introduces predictive coding as a unified algorithmic framework for constructing cognitive maps that can naturally extend to the mapping of auditory, sensorimotor, and linguistic inputs.

\end{abstract}

\begin{multicols}{2}
  \begin{topfigure}[t!]
  \centering
  \includegraphics[width=\textwidth]{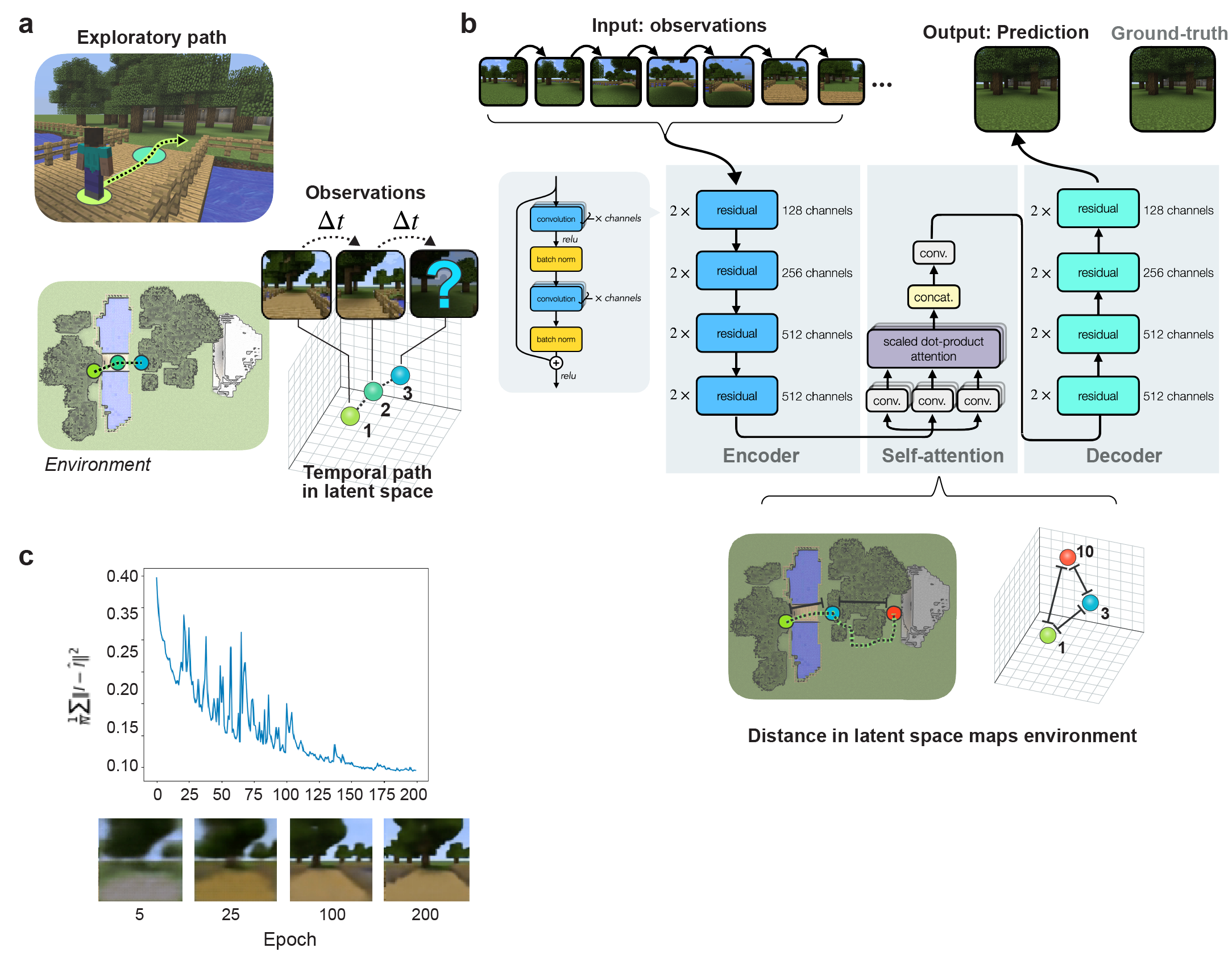}
  \caption{\textbf{A predictive coding neural network explores a virtual environment.} In
    predictive coding, a model predicts observations and updates its
    parameters using the prediction error. \textbf{a,} an agent's
traverses its environment by taking the most direct path to random positions. \textbf{b,} a self-attention-based
    encoder-decoder neural network architecture learns to perform predictive coding. A ResNet-18
    convolutional neural network acts as an encoder; self-attention is performed with 8 heads, and a
    corresponding ResNet-18 convolutional neural network performing decoding to the predicted
    image.\textbf{c,} the neural network learns to perform predictive coding effectively---with a
    mean-squared error of 0.094 between the actual and predicted images.
  }\label{fig:1}
\end{topfigure}

Space and time are fundamental physical structures in the natural world, and all organisms have evolved strategies for navigating space to forage, mate, and escape predation. \autocite{epsteinCognitiveMapHumans2017,wang2022localization,sivak2014environmental}. 
In humans and other mammals, the concept of a
spatial or cognitive map has been postulated to underlie spatial reasoning
tasks\autocite{andersonCognitivePsychologyIts2020,rescorla2009cognitive,whittington2022build}. A spatial map is an
internal, neural representation of an animal's environment that marks the location of landmarks, food, water, shelter, and then can be queried for navigation and
planning. The neural algorithms underlying spatial mapping are thought to  generalize to other sensory modes to provide cognitive representations of auditory and somatosensory data \autocite{aronovMappingNonspatialDimension2017b} as well as to construct internal maps of more abstract information including  concepts
\autocite{niehGeometryAbstractLearned2021,whittington2020tolman}, tasks \autocite{wilson2014orbitofrontal}, semantic information
\autocite{constantinescuOrganizingConceptualKnowledge2016a,garvertMapAbstractRelational2017,huthNaturalSpeechReveals2016b},
and memories \autocite{corkinLastingConsequencesBilateral1984}. Empirical evidence suggest that the brain uses common cognitive mapping strategies for spatial and non-spatial sensory information so that common mapping algorithms might exist that can map and navigate over not only visual but also semantic information and logical rules inferred from experience\autocite{behrens2018cognitive,aronovMappingNonspatialDimension2017b,niehGeometryAbstractLearned2021}. In such a paradigm reasoning itself could be implemented as a form of navigation within a cognitive map of concepts, facts, and ideas.

Since the notion of a spatial or cognitive map emerged, the question of how environments are represented within the brain and how the maps can be learned from experience has been a central question in neuroscience \autocite{okeefePlaceUnitsHippocampus1976}. Place cells in the hippocampus are neurons that are active when an animal transits through a specific location in an environment \autocite{okeefePlaceUnitsHippocampus1976}. Grid cells in the entorhinal cortex fire in regular spatial intervals and likely track an organism’s displacement in the environment \autocite{haftingMicrostructureSpatialMap2005a,amaralNeuronsNumbersHippocampal1990}. Yet with the identification of a substrate for the representation of space, the question of how a spatial map can be learned from sensory data has remained, and the neural algorithms that enable the construction of spatial and other cognitive maps remain poorly understood.

Empirical work in machine learning has demonstrated that deep neural networks can solve spatial navigation tasks as well as perform path prediction and grid cell formation \autocite{cuevaEmergenceGridlikeRepresentations2018,baninoVectorbasedNavigationUsing2018}. Cueva
\& Wei\autocite{cuevaEmergenceGridlikeRepresentations2018} and Banino \textit{et
  al.}\autocite{baninoVectorbasedNavigationUsing2018} demonstrate that neural networks can learn to
perform path prediction and that networks generate firing patterns that resemble the firing patterns
of grid cells in the entorhinal cortex. { 
Crane \textit{et al.}\autocite{craneHeatMethodDistance2017b}, Zhang \textit{et
al.}\autocite{Zhang2021.09.24.461751}, and Banino \textit{et
al.}\autocite{baninoVectorbasedNavigationUsing2018} demonstrate navigation
algorithms that require the environment's map or using firing patterns that resemble
place cells in the hippocampus.}
These studies allow an agent to access
environmental coordinates explicitly \autocite{cuevaEmergenceGridlikeRepresentations2018} or initialize a model
with place cells that represent specific locations in an
arena\autocite{baninoVectorbasedNavigationUsing2018}. In machine learning and autonomous navigation, a variety of algorithms have been developed to perform mapping tasks including SLAM and monocular SLAM algorithms
\autocite{thrunGraphSLAMAlgorithm2006,murartalVisualInertialMonocularSLAM2017,mourikisMultiStateConstraintKalman2007,lynenGetOutMy2015}
as well as neural network implementations
\autocite{guptaCognitiveMappingPlanning2019a,mirowskiLearningNavigateCities2018,duanRLFastReinforcement2016}. Yet,
SLAM algorithms contain many specific inference strategies, like visual feature
and object detection, that are specifically engineered for map building,
wayfinding, and pose estimation based on visual information.
{ Whereas extensive research in computer vision and machine learning use video
frames, these studies do not extract representations of the environment's map\autocite{higginsDARLAImprovingZeroShot2017,seoReinforcementLearningActionFree2022}.}
A unified theoretical and mathematical framework for understanding the mapping of spaces based on sensory information remains incomplete.

Predictive coding has been proposed as a unifying theory of neural function
where the fundamental goal of a neural system is to predict future observations
given past
data\autocite{lee2003hierarchical,mumford1994pattern,raoPredictiveCodingVisual1999b}. When
an agent explores a physical environment, temporal correlations in sensory
observations reflect the structure of the physical environment. Landmarks nearby
one another in space will also be observed in temporal sequence. In this way,
predicting observations in a temporal series of sensory observations requires an
agent to internalize some implicit information about a spatial
domain. Historically, Poincare motivated the possibility of spatial mapping
through a predictive coding strategy where an agent assembles a global
representation of an environment by gluing together information gathered through
local
exploration~\autocite{poincareFoundationsScienceScience2015,okeefeHippocampusCognitiveMap1978}.
The exploratory paths together contain information that could, in principle,
enable the assembly of a spatial map for both flat and curved manifolds.
{Indeed, extended Kalman filters\autocite{mourikisMultiStateConstraintKalman2007,thrunProbabilisticRobotics2005a}
for SLAM perform a form of
predictive coding by directly mapping visual changes and movement to spatial
changes. However, extended Kalman filters as well as other SLAM approachs
require intricate strategies for landmark size calibration, image feature
extraction, and models of the camera's distortion whereas biological systems can
solve flexible mapping and navigation that engineered systems cannot.}
Yet, while the concept of predictive coding for spatial mapping is intuitively
attractive, a major challenge is the development of algorithms that can glue
together local, sensory information gathered by an agent into a global,
internally consistent environmental map. Connections between mapping and
predictive coding in the literature have primarily focused on situations where
an agent has explicit access to its spatial location as a state variable
\autocite{stachenfeld2017hippocampus,recanatesi2021predictive,fang2023neural}. The
problem of building spatial maps \textit{de novo} from sensory data remains
poorly understood.

Here, we demonstrate that a neural network trained on a sensory, predictive
coding task can construct an implicit spatial map of an environment by
assembling observations acquired along local exploratory paths into a global
representation of a physical space within the network’s latent space. We analyze
sensory predictive coding theoretically and demonstrate mathematically that
solutions to the predictive sensory inference problem have a mathematical
structure that can naturally be implemented by a neural network with a
`path-encoder,' an internal spatial map, and a `sensory
decoder{,' and trained using backpropagation}. In such a
paradigm, a network learns an internal map of its environment by inferring an
internal geometric representation that supports predictive sensory inference. We
implement sensory predictive coding within an agent that explores a virtual
environment while performing visual predictive coding using a convolutional
neural network with self-attention. Following network training during
exploration, we find that the encoder network embeds images collected by an
agent exploring an environment into an internal representation of space. Within
the embedding, the distances between images reflect their relative spatial
position, not object-level similarity between images. During exploratory
training, the network implicitly assembles information from local paths into a
global representation of space as it performs a next image inference
problem. Fundamentally, we connect predictive coding and mapping tasks,
demonstrating a computational and mathematical strategy for integrating
information from local measurements into a global self-consistent environmental
model.

  \begin{topfigure}[t!]  \centering \includegraphics[width=\textwidth]{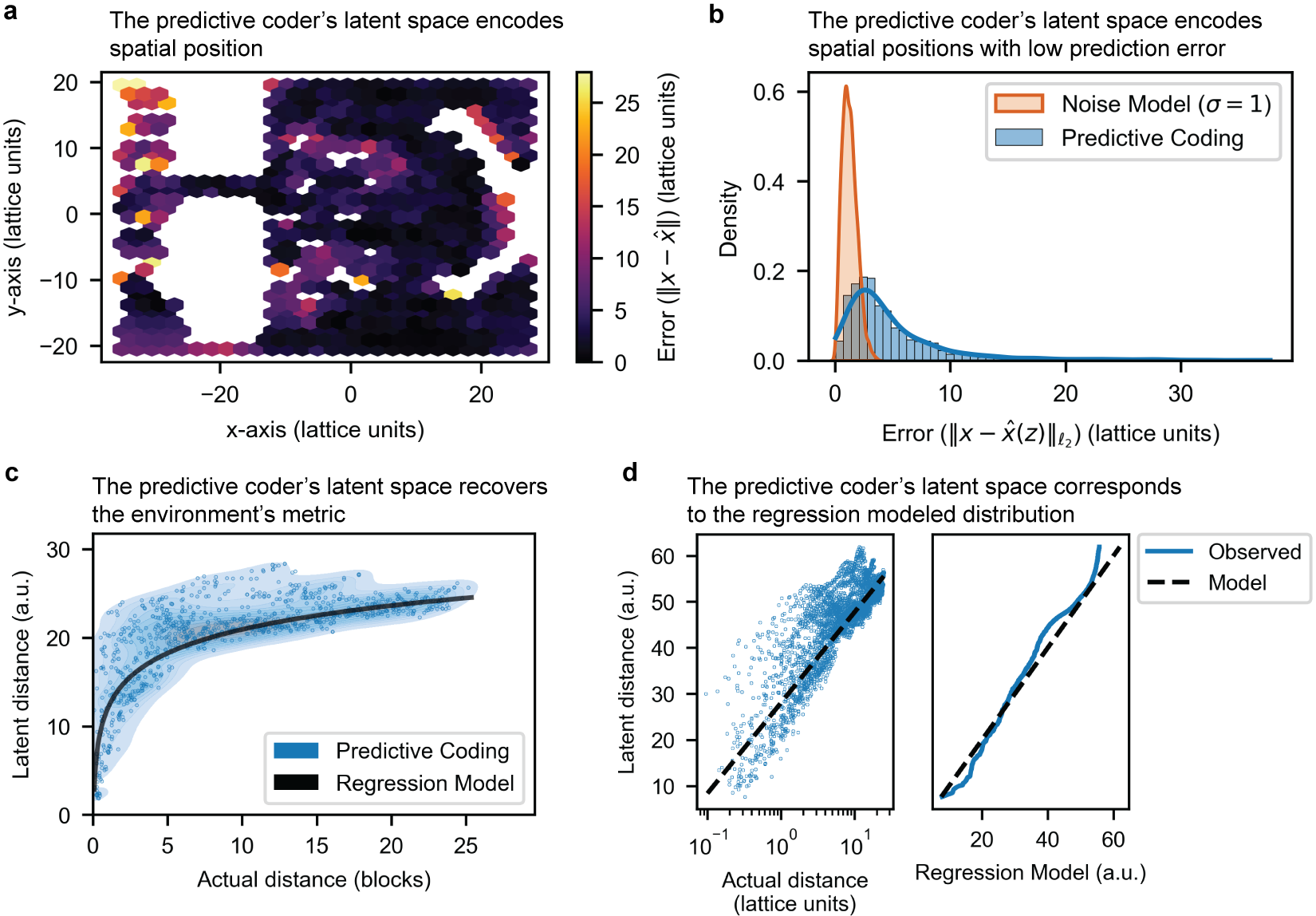}
 \caption{
\textbf{Predictive coding neural network constructs an implicit spatial map.}
{ \textbf{a-b,} The predictive coder's latent space encodes accurate spatial positions. A neural
network predicts the spatial location from the predictive coding’s latent space. \textbf{a,} a
heatmap of the prediction errors between the actual position and the predictive coder's predicted
positions show a low prediction error. \textbf{b,} The histogram of prediction errors of positions from the predictive coder's
latent space show a low prediction error. As a baseline (Noise model
($\sigma = 1$ lattice unit)), actual positions with a small noise
displacement gives an error model. \textbf{c,} predictive coding’s latent distances recover the environment’s spatial
metric. Sequential visual images are mapped to the neural network's latent space, and the latent
space distances ($\ell_2$) are plotted with physical distances onto a joint density
plot. An nonlinear regression model
\begin{minipage}{\linewidth}
\[
    \norm{z - z'} = \alpha \log \norm{x - x'} + \beta
\]
\end{minipage}
is shown as a baseline. \textbf{d,} a correlation plot and a quantile-quantile plot show the
overlap between the empirical and model distributions.}
}\label{fig:2} \end{topfigure}

\section*{Mathematical formulation of spatial mapping as sensory predictive coding}\label{sec:2}
In this paper, we aim to understand how a spatial map can be assembled by an
agent that is making sensory observations while exploring an environment. Papers
in the literature that study connections between predictive coding and mapping
have primarily focused on situations where an agent has access to its `state' or
location in the
environment\autocite{stachenfeld2017hippocampus,recanatesi2021predictive,fang2023neural}.
Here, we develop a theoretical model and neural network
implementation\footnote{ The neural network is a feedforward deep
neural network trained using backpropagation, or gradient descent,
rather than Helmholtz machines\autocite{dayanHelmholtzMachine1995,luttrellBayesianAnalysisSelfOrganizing1994}, which are commonly used in predictive
coding.} of
sensory predictive coding that illustrates why and how an internal spatial map
can emerge naturally as a solution to sensory inference problems. We, first,
formulate a theoretical model of visual predictive coding and demonstrate that
the predictive coding problem can be solved by an inference procedure that
constructs an implicit representation of an agent's environment to predict
future sensory observations.  The theoretical analysis also suggests that the
underlying inference problem that can be solved by an encoder-decoder neural
network that infers spatial position based upon observed image sequences.

We consider an agent exploring an environment, $\Omega \subset \mathbb{R}^2$,
 while acquiring visual information in the form of pixel valued image vectors
$I_x \in \mathbb{R}^{m \times n}$ given an $x \in \Omega$.  The agent's
environment $\Omega$ is a bounded subset of $\mathbb{R}^2$ that could contain
obstructions and holes. In general, at any given time, $t$, the agent's state
can be characterized by a position $x(t)$ and orientation $\theta(t)$ where
$x(t)$ and $\theta(t)$ are coordinates within a global coordinate system unknown
to the agent.  

The agent's environment comes equipped with a visual scene, and the agent makes
observations by acquiring image vectors $I_{x_k} \in \mathbb{R}^{m \times n}$ as
 it moves along a sequence of points $x_k$. At every position $x$ and
 orientation $\theta$, the agent acquires an image by effectively sampling from
an image the conditional probability distribution $P(I|x_k, \theta_k)$ which
encodes the probability of observing a specific image vector $I$ when the agent
is positioned at position $x_k$ and orientation $\theta_k$. The distribution
$P(I|x, \theta)$ has a deterministic
and stochastic component where the deterministic component is set by landmarks
in the environment while stochastic effects can emerge due to changes in
lighting, background, and scene dynamics. { Mathematically, we can view $P(I|x, \theta)$
as a function on a vector bundle with base space $\Omega$ and total space
$\Omega \times I$\autocite{tuDifferentialGeometryConnections2017a}.  The function assigns an observation probability to every
possible image vector for an agent positioned at a point $(x,
\theta)$. Intuitively, the agent's observations preserve the geometric
structure of the environment: the spatial structure influences
temporal correlations.}

In the predictive coding problem, the agent moves along a series of points
$(x_0, \theta_0),
(x_1, \theta_1), \dots, (x_k, \theta_k)$ while acquiring images $I_0, I_1, \dots I_k$.  The motion of
the agent in $\Omega$ is generated by a Markov process with transition
probabilities $P(x_{i+1}, \theta_{i+1}|x_i, \theta_i)$. Note that the agent has access to the image
observations $I_i$ but not the spatial coordinates $(x_i, \theta_i)$. Given the set $\{ I_0
\dots I_k\}$ the agent aims to predict $I_{k+1}$. Mathematically, the image
prediction problem can be solved theoretically through statistical inference by
(a) inferring the posterIor probability distribution
$P(I_{k+1}|I_0,I_1....I_{k})$ from observations. Then, (b) given a specific
sequence of observed images $\{ I_0 \dots I_k\}$, the agent can predict the next
image $I_{k+1}$ by finding the image $I_{k+1}$ that maximizes the posterior
probability distribution $P(I_{k+1}|I_0,I_1....I_{k})$.

The posterior probability  distribution $P(I_{k+1}|I_0,I_1,....I_{k})$ is by definition

\begin{align*}
  & P(I_{k+1}|I_0,I_1,....,I_{k}) = \frac{P(I_0,I_1,...,I_{k},I_{k+1})}{P(I_0,I_1,...,I_k) }. \\
  \end{align*}

If we consider $P(I_0,I_1 \dots I_k,I_{k+1})$ to be a function of an
implicit set of spatial coordinates $(x_i, \theta_i)$ where the $(x_i,
\theta_i)$ provide an internal representation of the spatial environment. Then, we can express the posterior probability $P(I_{k+1}|I_0,I_1....I_{k})$ in terms of the implicit spatial representation 

\begin{widetext}
\begin{align}\label{eq:path-int}
  P(I_{k+1}|I_0,I_1,\ldots,I_{k})
  &=   \int_{\Omega} \mathbf{dx} \, \mathbf{d\theta} \ P(x_0, \theta_0, x_1, \theta_1,
    \dots, x_k, \theta_k) \frac{P(I_0,I_1
    \cdots I_k|x_0, \theta_0, \dots, x_k, \theta_k)}{P(I_0,I_1, \dots,I_k) }
    P(x_{k+1}|x_k, \theta_k) P(I_{k+1}|x_{k+1}, \theta_{k+1}) \nonumber \\
 &=   \int_{\Omega} \ \mathbf{dx} \, \mathbf{d\theta} \
   \underbrace{P(x_0, \theta_0,x_1, \theta_1, ...,x_k, \theta_k|I_0,
   I_1 \dots, I_k)}_{\text{encoding} (1)} \underbrace{P(x_{k+1},
   \theta_{k+1}|x_k, \theta_k)}_{\text{spatial transition probability } (2)}
   \underbrace{P(I_{k+1}|x_{k+1}, \theta_{k+1})}_{\text{decoding } (3)}
\end{align}
\end{widetext}
where in \refeq{eq:path-int} the integration is over all possible paths
 $\{(x_0, \theta_0) \dots (x_k, \theta_k) \}$ in the domain $\Omega$, $\mathbf{dx} = dx_0 \dots
dx_k$, and $\mathbf{d\theta} = d\theta_0 \cdots d\theta_k$. \refeq{eq:path-int} can be interpreted
as a path integral over the domain $\Omega$.  The path integral assigns a probability to every
possible path in the domain and then computes the probability that the agent will observe a next
image $I_k$ given an inferred location $(x_{k+1}, \theta_{k+1})$. In detail, term (1) assigns a probability to every
discrete path $\{(x_0, \theta_0) \dots (x_k, \theta_k)\} \in \Omega$ as the conditional likelihood of the path given the
observed sequences of images $\{I_0 \dots I_k\}$ .  Term (2) computes the probability that an agent at
a terminal position $x_k$ moves to the position $(x_{k+1}, \theta_{k+1})$ given the Markov transition function
$P(x_{k+1}, \theta_{k+1}|x_k, \theta_k)$. Term (3) is the conditional probability that image $I_{k+1}$ is observed given that
the agent is at position $(x_{k+1}, \theta_{k+1})$.

Conceptually, the product of terms solves the next image prediction problem in three steps. First
(1), estimating the probability that an agent has traversed a particular sequence of points given
the observed images; second (2), estimating the next position of the agent $(x_{k+1}, \theta_{k+1})$ for each
potential path; and third (3), computing the probability of observing a next image $I_{k+1}$ given
the inferred
terminal location $x_{k+1}$ of the agent. Critically, an algorithm that implements the inference
procedure encoded in the equation would construct an internal but implicit representation of the
environment as a coordinate system $\mathbf{x}, \bm{\theta}$ that is learned by the agent and used during the
next image inference procedure.  The coordinate system provides an internal, inferred representation
of the agent's environment that is used to estimate future
image observation probabilities. Thus, our theoretical framework demonstrates how an agent might construct an implicit representation of its spatial environment by solving the predictive coding problem.

The three step inference procedure represented in the equation for $P(I_{k+1}|I_0 \dots I_k)$ can be
directly implemented in a neural network architecture, as demonstrated
in the \hyperlink{sec:app-nn}{Appendix}. The first term acts as an `encoder' network
that computes the probability that the agent has traversed a path
$\{(x_0, \theta_0) \dots (x_k, \theta_k)\}$
given an observed
image sequence
$I_0,\ldots, I_k$ that has been observed by the network (\reffig{fig:1}{(\textbf{b})}). The network
can, then, estimate the next position of the agent $(x_{k+1},
\theta_{k+1})$ given an inferred location $(x_k, \theta_k)$, and
apply a decoding network to compute $P(I_{k+1}|x_{k+1}, \theta_{k+1})$ while outputting the prediction $I_{k+1}$
using a decoder.  A network trained through visual experience must learn an internal coordinate
system and representation $\mathbf{x}, \bm{\theta}$ that not only offers an environmental representation but also
establishes a connection between observed images $I_j$ and inferred
locations $(x_j, \theta_j)$.

\subsection*{A neural network performs accurate predictive coding within a virtual environment}

Motivated by the implicit representation of space contained in the predictive coding inference problem, we developed a computational implementation of a predictive coding agent, and studied the representation of space learned by that agent as it explored a virtual environment. We first create an environment with the Malmo environment in Minecraft \autocite{Johnson2016TheMP}. The physical environment measures 40 × 65 lattice units and encapsulates three aspects of visual scenes: a cave provides a global visual landmark, a forest provides degeneracy between visual scenes, and a river with a bridge constrains how an agent traverses the environment (\reffig{fig:1}{(\textbf{a})}). An agent follows paths (\reffig{fig:supp-E}{(\textbf{b-c})}), determined by $A^*$ search, between randomly sampled positions and receives visual images along every path. 

To perform predictive coding, we construct an encoder-decoder convolutional neural network (CNN) with a
ResNet-18 architecture\autocite{heDeepResidualLearning2015} for the encoder and a corresponding
ResNet-18 architecture with transposed convolutions in the decoder (\reffig{fig:1}{(\textbf{b})}). The
encoder-decoder architecture uses the U-Net
architecture\autocite{ronnebergerUNetConvolutionalNetworks2015a} to pass the encoded latent units into
the decoder. Multi-headed attention\autocite{vaswaniAttentionAllYou2023} processes the sequence of
encoded latent units to encode the history of past visual observations. The multi-headed attention
has $h = 8$ heads. For the encoded latent units with dimension $D = C \times H \times W$, the
dimension $d$ of a single head is $d = C \times H \times W / h$.

The predictive coder approximates predictive coding by minimizing the mean-squared error between the
actual observation and its predicted observation. The predictive coder trains on $82,630$ samples
for $200$ epochs with gradient descent optimization with Nesterov
momentum\autocite{sutskeverImportanceInitializationMomentum}, a weight decay of $5 \times 10^{-6}$,
and a learning rate of $10^{-1}$ adjusted by OneCycle learning rate
scheduling\autocite{smithSuperConvergenceVeryFast2018}. The optimized predictive coder has a
mean-squared error between the predicted and actual images of $0.094$ and a good visual fidelity
(\reffig{fig:1}{(\textbf{c})}).

  \begin{figure*}
  \centering
  \includegraphics[width=\textwidth]{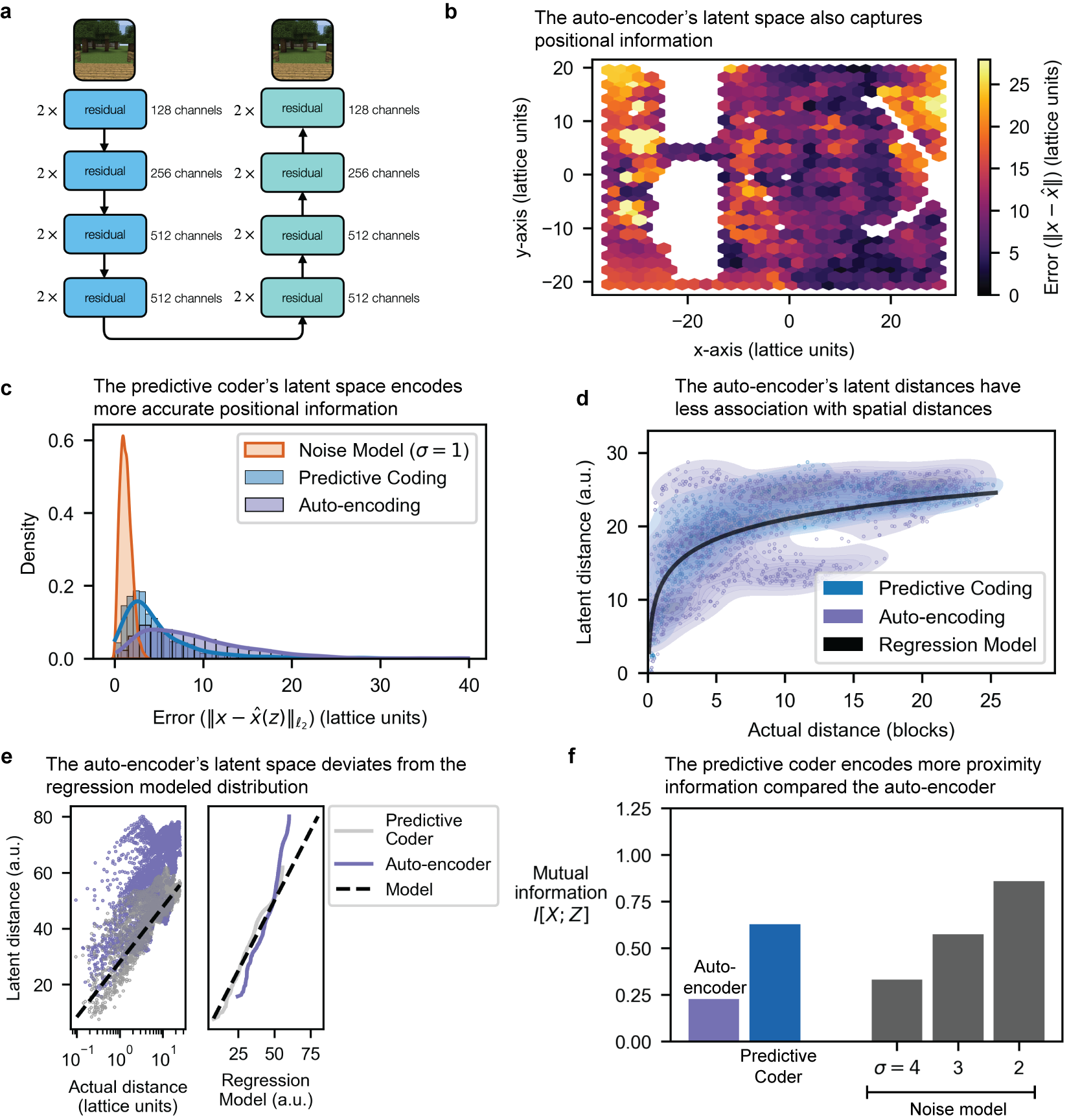}
  \caption{\textbf{Predictive coding network learns spatial proximity
not image similarity} { \textbf{a,} an autoencoding neural network compresses
    visual images into a low-dimensional latent vector and reconstructs the image from the latent
    space. Auto-encoder trains on visual images from the environment \textit{without any sequential
      order}. \textbf{b-c,} auto-encoding encodes lower resolution in positional
    information. \textbf{b,} a neural network predicts the spatial location from the auto-encoding’s
    latent space. A heatmap of the prediction errors between the
    actual position and the auto-encoder's predicted positions
    show a higher prediction error---compared to the predictive coder. \textbf{c,} auto-encoding captures
less positional information compared to
    predictive coding. The histogram shows the prediction errors of positions from the latent space
    of both the auto-encoder and the predictive coder. \textbf{d,} latent distances, however, show a
    weaker relationship with physical distances, as the joint histogram between physical and latent
    distances is less concentrated. \textbf{e,} a correlation plot and
a quantile-quantile plot show a lower correlation and a lower density overlap
 between the empirical and model distributions.
  \textbf{f,} predictive coding’s latent units communicate more
    fine-grained spatial distances whereas auto-encoding communicates broad spatial regions. Joint
    density plots show the association between latent distances and physical distances for both
    predictive coding and auto-encoding. Predictive coding’s latent distances increase with spatial
    distances, with a higher concentration compared to auto-encoding.}}\label{fig:3}
\end{figure*}

\section*{Predictive coding network constructs an implicit spatial map} \label{sec:map}

We show that the predictive coder creates an implicit spatial map by demonstrating it recovers the
environment's spatial position and distance. We encode the image sequences using the predictive
coder's encoder to analyze the encoded sequence as the predictive coder's latent units. To measure
the positional information in the predictive coder, we train a neural network to predict the agent's
position from the predictive coder's latent units (\reffig{fig:1}{(\textbf{a})}). The neural network's
prediction error
\[
  E(x, \hat{x}) = \norm{\hat{x} - x}_{\ell_2}
\]
indirectly measures the predictive coder's positional
information. { To provide comparative baselines, we construct a
position prediction model.  To lower bound the prediction error, we
construct a model that gives the agent's actual position with small additive Gaussian noise
\[
  \hat{x} = x + \epsilon, \epsilon \sim \mathcal{N}(0, \sigma).
\]
To compare the predictive coder to the baselines, we compare the
prediction error histograms (\reffig{fig:2}{(\textbf{b})}).
}

The predictive coder encodes the environment's spatial position to a low prediction error
(\reffig{fig:2}{.d}). The predictive coder has a mean error of $5.04$ lattice units and $> 80\%$ of
samples have an error $< 7.3$ lattice units. The additive Gaussian model with $\sigma = 4$ has a
mean error of $4.98$ lattice units and $>80\%$ of samples with an
error $< 7.12$ lattice units.

We show the predictive coder's latent space recovers the local distances between the environment's
physical positions. For every path that the agent traverses, we calculate the local pairwise
distances in physical space and in the predictive coder's latent space with a neighborhood of 100
time points. To determine whether latent space distances correspond to physical distances, we
calculate the joint density between latent space distances and physical distances
(\reffig{fig:2}{(\textbf{c})}). We model the latent distances by fitting the physical distances with
additive Gaussian noise to a logarithmic function
\[
  d(z, z') = \alpha \log (\norm{x - x' + \epsilon}) + \beta, \epsilon \sim \mathcal{N}(0, \sigma).
\]
The modeled distribution is concentrated with the predictive coder's
distribution (\reffig{fig:2}{(\textbf{d})}) with a { Pearson correlation coefficient
of 0.827 and a} Kullback-Leibler divergence $(\mathbb{D}_\text{KL}(p_\text{PC} \Vert
p_\text{model}))$ of $0.429$ bits.

  \begin{figure*}
  \centering
  \includegraphics[width=\textwidth]{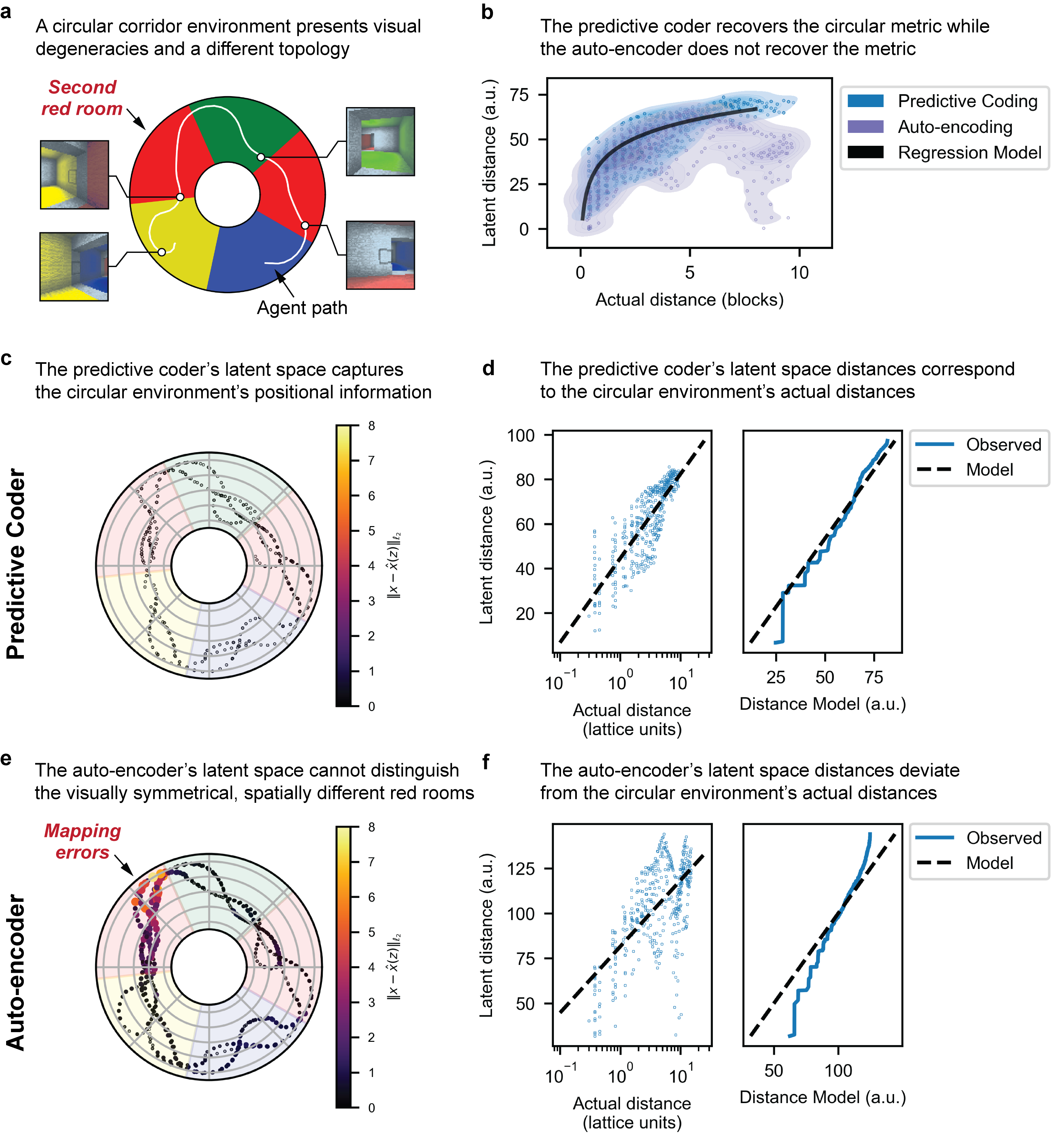}
  \caption{\textbf{Predictive coding network can learn a circular
topology and distinguishes visually
identical, spatially different locations} { \textbf{a,} an agent traverses a
circular environment with two visually identical red rooms, which
provides visually similar yet spatially different
locations. \textbf{b,} the predictive coder's latent distances show a
correspondence with the circular environment's metric while the
auto-encoder's latent distances show little correlation. \textbf{c,}
similar to Figures 2 and 3, a different neural network measures the
predictive coder's spatial information by predicting the agent's location from the
predictive coder's latent space. The predictive coder's latent space
demonstrates a low prediction error. \textbf{d,} similar to Figures 2 and 3, the nonlinear
regression measures the correspondence between the latent distances
$\norm{z - z'}$ and the actual distances $\norm{x - x'}$ with the model
\begin{minipage}{\linewidth}
\[
    \norm{z - z'} = \alpha \log \norm{x - x'} + \beta.
\]
\end{minipage}
\textbf{d,} the correlation plot (left) with
the nonlinear regression model show a strong correlation between the
predictive coder's latent distances and the environment's actual distances
($r = 0.827$). The quantile-quantile plot (right) between the
predictive coder's latent distances and the regression model show high
overlap ($\mathbb{D}_\text{KL}(p_\text{PC} \Vert p_\text{model}) =
0.250$). 
  \textbf{e,} without any past information, the auto-encoder cannot
distinguish the two different red rooms and produces a high prediction
error in these locations. \textbf{f,} the correlation plot (left) with
the nonlinear regression model show little correlation between the
auto-encoder's latent distances and the environment's actual distances
($r = 0.288$). The quantile-quantile plot (right) between the
auto-encoder's latent distances and the regression model show little
overlap ($\mathbb{D}_\text{KL}(p_\text{PC} \Vert p_\text{model}) =
3.806$).}  }\label{fig:3.1}
\end{figure*}

\section*{Predictive coding network learns spatial proximity not image similarity} \label{sec:auto}

In the \hyperref[sec:map]{previous section}, we show that a neural network that performs predictive
coding learns an internal representation of its physical environment within its latent space. Here, we demonstrate that the prediction task itself is essential for spatial mapping . Prediction forces a network to learn spatial proximity and not merely image similarity.  Many frameworks including principal
components analysis, IsoMap\autocite{tenenbaumGlobalGeometricFramework2000}, and autoencoder neural networks can collocate images by visual similarity.  While similar scenes might be proximate in space, similar scenes can also be spatially divergent. For example, the virtual environment we constructed has two different `forest' regions that are separated by a lake. Thus, in the two  forest environments might generate similar images but are actually each closer to the lake region than to one another (\reffig{fig:1}{(\textbf{a})}).

To demonstrate the central role for prediction in mapping, we compared the latent representation of images generated by the predictive coding network to a representation learned by an autoencoder. The
auto-encoder network has a similar architecture to the predictive encoder but encodes a \textit{single} image observation in a latent space, and decodes the same observations.  As the auto-encoder only operates on a single image---rather than a sequence, the
auto-encoder learns an embedding based on image proximity not underlying spatial relationships. As with the predictive coder, the auto-encoder (\reffig{fig:3}{(\textbf{a})}) trains to minimize the
mean-squared error between the actual image and the predicted image on $82,630$ samples for $200$
epochs with gradient descent optimization with Nesterov momentum, a weight decay of $5 \times
10^{-6}$, and a learning rate of $10^{-1}$ adjusted by the OneCycle learning rate scheduler. The
auto-encoder has mean-squared error of $0.039$ and a high visual fidelity.

{ The predictive coder encodes a higher resolution and more accurate spatial map
in its latent space than the auto-encoder.} As
with the predictive coder, we train an auxiliary neural network to
predict the agent's position from the auto-encoder's latent units
(\reffig{fig:3}{(\textbf{b})}). The neural network's prediction error
indirectly measures the auto-encoder positional information. For
greater than $80\%$ of the auto-encoder's points, its prediction error
is less than $13.1$ lattice units, as compared to the predictive coder that has $> 80\%$ of its samples below a prediction error of $7.3$ lattice units (\reffig{fig:3}{(\textbf{c})}).

We also show that the predictive coder recovers the environment's spatial distances with finer
resolution compared to the auto-encoder. As with the predictive coder, we calculate the local pairwise
distances in physical space and in the auto-encoder's latent space, and we generate the joint
density between the physical and latent distances (\reffig{fig:3}{(\textbf{d})}). Compared to the predictive
coder's joint density, the auto-encoder's latent distances increase with the agent's physical
distance. The auto-encoder's joint density shows a larger dispersion compared to the predictive coder's
joint density, indicating that the auto-encoder encodes spatial distances with higher uncertainty.

We can quantitatively measure the dispersion in the auto-encoder's joint density by calculating
mutual information of the joint density (\reffig{fig:3}{(\textbf{e})})
\[
  I[X; Z] = \mathbb{E}_{p(X, Z)}\left[ \log \frac{p(X, Z)}{p(X)p(Z)} \right].
\]
The auto-encoder has a mutual information of $0.227$ bits while the predictive coder has a mutual
information of $0.627$ bits. As a comparison, positions with additive Gaussian noise having a standard
deviation $\sigma$ of $2$ lattice units has a mutual information of $0.911$ bits. The predictive
coder encodes $0.400$ additional bits of distance information to the auto-encoder. The predictive
coder's additional distance information of $0.4$ bits exceeds the auto-encoder's distance
information of $0.227$ bits, which indicates the temporal dependencies encoded by the predictive
coder capture more spatial information compared to visual similarity.

{ \section*{Predictive coding network maps visual
degenerate environments whereas auto-encoding cannot}

The sequential prediction task is beneficial for spatial mapping: the predictive coder captures more
accurate spatial information compared to the auto-encoder, and predictive coder's latent distances
have a stronger correspondence to the environment's metric. However, it is unclear whether
predictive coding is \textit{necessary} (as opposed to \textit{beneficial}) to recover an
 environment's map; an auto-encoder may still recover the environment's map. In this section, we
demonstrate that predictive coding is necessary for recovering an environment's map. First, we show
empirically that there exist environments that auto-encoding cannot recover. Second, we provide
insight into why the auto-encoder fails with a theorem showing that auto-encoding cannot recover
many environments---specifically, environments with visually similar yet spatially different
locations.

In the previous sections, the agent explores a natural environment with forest, river, and cave
landmarks. While this environment models exploration in outdoor environments, the lack of controlled
visual scenes complicates interpreting predictive coder and auto-encoder. We introduce a circular
corridor (\reffig{fig:3.1}{(\textbf{a})}) to introduce visual scenes that \textit{visually
identical}---rather than \textit{visually similar}---yet spatially
different. Specifically, the rooms revolve clockwise as red, green,
red, blue, and yellow; there exist two distinct red rooms. The two
distinct red rooms answers two questions: can the the predictive coder
and auto-encoder recover the map for
environments with visual symmetry, and does the predictive coder recover a global map or a relative
map? In other words, does the predictive coder recovers the circular corridor's geometry, or does it
learn a linear hallway?

Similar to previous sections, we train a neural network (or a predictive coder) to perform
predictive coding while traversing the circular corridor. In addition, we train a neural network (or
an auto-encoder) to perform auto-encoding. The auto-encoder fails to recover spatial information in areas with visual
degeneracy: it maps the two distinct red rooms to the same location
(\reffig{fig:3.1}{(\textbf{e})}). In \reffig{fig:3.1}{(\textbf{e})},
the auto-encoder predicts the images in the left red room to locations
in the right red room---whereas locations with distinct visual scenes
(such as the yellow and blue rooms) show a low prediction
error (mean error $\norm{x - \hat{x}}_{\ell_2} = 5.004$ lattice
units). In addition, the auto-encoder's latent distances do not
separate the different red rooms in latent space---whereas the
predictive coder separates the two red rooms (\reffig{fig:3.1}{(\textbf{b})}).
Moreover, the predictive coder demonstrates a low prediction
error throughout, including the two visually degenerate red rooms
(mean error $\norm{x - \hat{x}}_{\ell_2} = 0.071$ lattice units) (\reffig{fig:3.1}{(\textbf{c})}).

Moreover, we measure the
relationship between the predictive coder's (and auto-encoder's) metric and the environment's metric
by fitting a regression model (\reffig{fig:3.1}{(\textbf{b})})
\[ 
  \norm{z - z'} = \alpha \log \norm{x - x'} + \beta 
\] 
between the predictive coder's (and auto-encoder's) latent distances
($\norm{z - z'}$) and the environment's physical distances ($\norm{x - x'}$). Compared to the
natural environment, the auto-encoder's latent distances shows more deviation from the environment's
spatial distances whereas the predictive coder's latent distances maintain a correspondence with
spatial distances. For the predictive coder, the latent metric recovers spatial metric
quantitatively: the correlation plot (\reffig{fig:3.1}{(\textbf{d}, left)}) shows a high correlation
($r = 0.827$) between the latent and spatial distances, and the quantile-quantile plot
(\reffig{fig:3.1}{(\textbf{d}, right)}) shows a high overlap between the regression model and the
observed latent distances ($\kl (p_\text{PC} \Vert p_\text{model}) = 0.250$).  The auto-encoder's
latent metric, conversely, does not recover the spatial metric: the correlation plot
(\reffig{fig:3.1}{(\textbf{f}, left)}) shows a low correlation ($r = 0.288$) between the latent and
spatial distances, and the quantile-quantile plot (\reffig{fig:3.1}{(\textbf{f}, right)}) shows a
low overlap between the regression model and the observed latent distances ($\kl (p_\text{PC} \Vert
p_\text{model}) = 3.806$).

As shown in \reffig{fig:3.1}, the auto-encoder cannot recover the spatial map of the circular
corridor---whereas the predictive coder can recover the map. Here we show that auto-encoders cannot
the environment's map for \textit{any} environment with visual degeneracy, not just the circular
corridor. To show that the auto-encoder cannot learn the environment's map, we show that \textit{any}
statistical estimator cannot learn the environment's map from \textit{stationary} observations.
For clarity and brevity, we will provide a proof sketch on a lattice
environment $X$---a
closed subset of $\mathbb{Z}^2$.
\begin{theorem}\label{th:ae}
    Consider an environment $X$---a closed subset of the lattice
    $\mathbb{Z}^2$ with a function $x \xmapsto{f} I$ that gives an
image $I_x = f(x) \subset \mathbb{R}^D $ for each position $x \in
X$. Let the environment's observations be degenerate such that 
\[f(x_1) = f(x_2) \text{ for some } x_1 \neq x_2.\]
There exists no decoder $I \xmapsto{d} x$ that satisfies
\[
    x = d \circ I_x = d \circ f(x).
\]
\end{theorem}

\begin{proof}
    The proof proceeds as a consequence that a function has no
    left inverse \textit{if and only if} it is not one-to-one.
    
    Suppose there exists a decoder $I \xmapsto{d} x$ that satisfies
\[
    x = d \circ I_x = d \circ f(x).
\]
Consider
\[ f(x_1) = f(x_2) = I \text{ for some } x_1 \neq x_2.\]
Then, 
\begin{align*}
  x_1 = d \circ f(x_1) = d(I) = d \circ f(x_2) = x_2,
\end{align*}
which is a contradiction, as required.
\end{proof}
Because \refthm{th:ae} demonstrates there exists no decoder for a visually
degenerate environment with \textit{stationary} observations, an
auto-encoder \textit{cannot} recover a visually degenerate
environment; the auto-encoder's failure arises because two
locations with the same observation cannot be discriminated. 
\begin{corollary}
    Consider an auto-encoder $g = \text{dec} \circ \text{enc}$ with an
encoder $I \xmapsto{\text{enc}} z$ and decoder $z \xmapsto{\text{dec}}
I$ that compresses images into a latent space $z \in
\mathbb{R}^H$. There exists no decoder $z \xmapsto{h} x$ that
satisfies
\[
    x = h \circ z_x = h \circ \text{enc} \circ f(x).
\]
\end{corollary}
\begin{proof}
    Consider the decoder $d = h \circ \text{enc}: I \to x$. By \refthm{th:ae},
this decoder cannot satisfy
\[
    x = d \circ f(x) = h \circ \text{enc} \circ f(x),
\]
as required.
\end{proof}
}

  \begin{figure*}[p!]
  \centering
  \includegraphics[width=\textwidth]{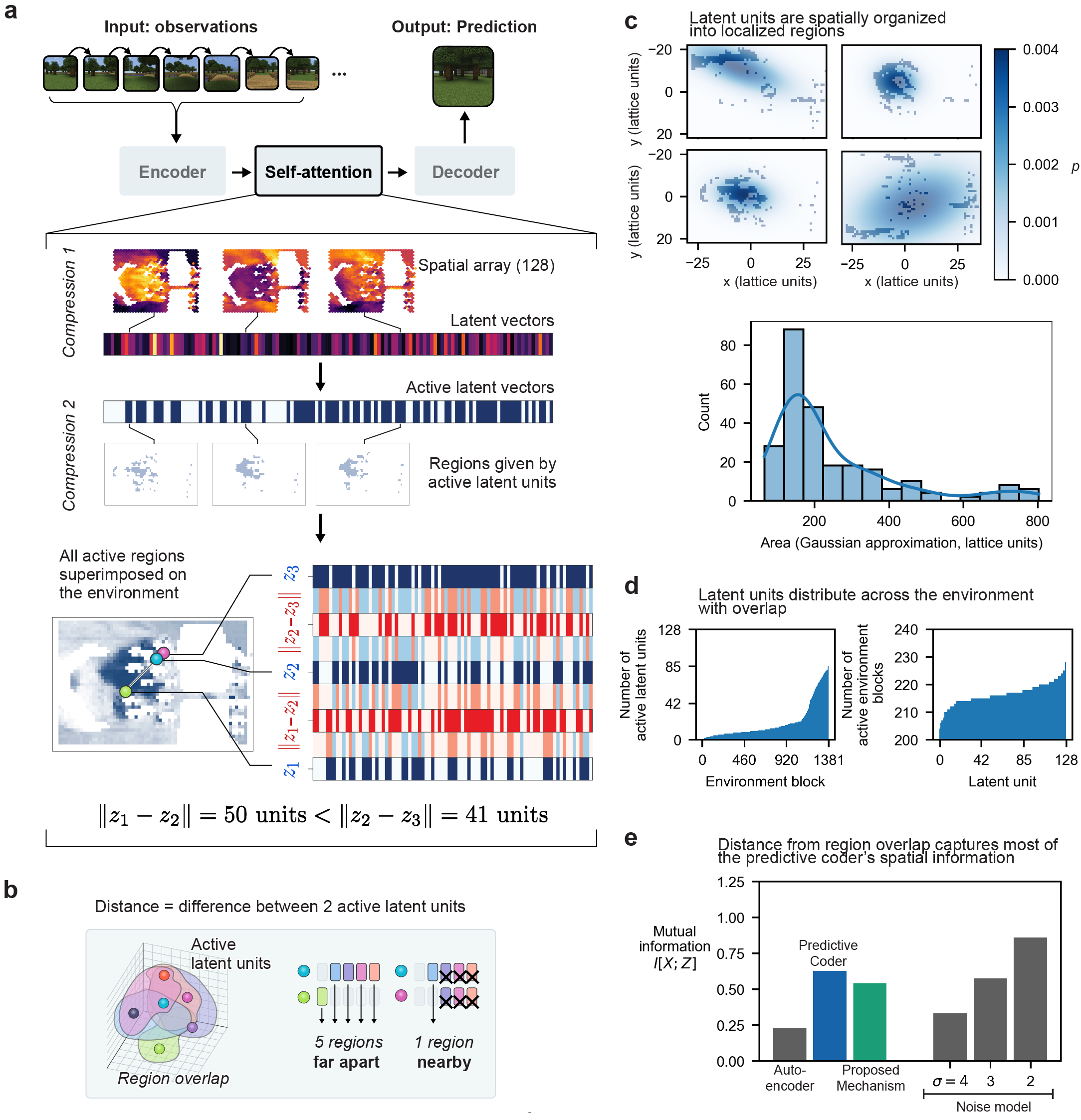}
  \caption{\textbf{The predictive coding  network generates place fields that support vector based distance calculations} \textbf{a,} when encoding past images for predictive coding, the self-attention
    module generates latent vectors. Each continuous unit in these latent vectors activates in
    concentrated, localized regions in physical space. These continuous units can be thresholded to
    generate a binary vector determining whether each unit is active. Each latent unit covers a
    unique region, and each physical location gives a unique combination of these overlapping
    regions. As an agent moves away from its original location, the combination of overlapping
    regions gradually deviates from its original combinations. This deviation, as measured by
    Hamming distance, correlates with physical distance. \textbf{b,} distance is given by the
    difference in the latent units' overlapping regions. Two nearby locations have small deviations
    in overlap (right) while two distant locations have large deviations (middle). \textbf{c}, latent
    units are spatially organized into localized regions. The active latent units are approximated
    by a two-dimensional Gaussian distribution to measure the latent unit's localization (top). The
    latent units' Gaussian approximations are highly localized with a mean area of $254.6$ for
    densities $p \geq 0.0005$. \textbf{d,} latent units distributed across the environment. The
    number of latent units was calculated as each lattice block in the environment (left), and the
    number of lattice blocks were calculated for each active unit (right). The latent units provide
    a unique combination for $87.6\%$ of the environment, and their aggregate covers the entire
    environment. \textbf{e,} distance from the region overlap captures most of the predictive
    coder's spatial information. We calculate the distance for every pair of active latent vectors
    and their respective physical Euclidean distances as a joint distribution. The proposed
    mechanism captures a majority of the predictive coder's spatial information---as the proposed
    mechanism's mutual information (0.542 bits) compares to the predictive coder's mutual
    information (0.627 bits)}\label{fig:4}
\end{figure*}

\section*{Predictive coding generates units with localized receptive fields that support vector navigation}

In the previous section, we demonstrate that the predictive coding neural network captures spatial relationships within an environment containing more internal spatial information than can be captured by an auto-encoder network that encodes image similarity. Here, we analyze the structure of the spatial code learned by the predictive coding network. We demonstrate that each unit in the neural network's latent space activates at
distinct, localized regions---akin to place fields in the mammalian brain---in the environment's physical space
(\reffig{fig:4}{(\textbf{a})}). These place fields overlap and their aggregate
covers the entire physical space. Each physical location, is represented by a unique combination of
overlapping regions encoded by the latent units. This combination of overlapping regions 
recovers the agent's current physical position. Furthermore, given two
physical locations, there now exist two distinct combinations of
overlapping regions in latent space.
{ Vector navigation is
the representation of the vector heading to a goal location from a
current location~\autocite{bushUsingGridCells2015b}. We show that overlapping regions (or place fields) can give a
heading from a current locations to a goal location. Specifically, a
linear decoder recovers the vector to a goal location from a starting
location by taking the difference in place fields, which supports
vector navigation\footnote{Studies such as Bush \textit{et al.} (2015)
consider grid cell-supported vector navigation whereas we only consider vector
navigation using place cells.} (\reffig{fig:supp-A}).}

To support this proposed mechanism, we first demonstrate the neural network generates place
fields. In other words, units from the neural network's latent space produce localized regions in
physical space. To determine whether a latent unit is active, we threshold the continuous value with
its 90th-percentile value. { The agent has head
direction varies to ensure the regions are stable across all head directions.}
To measure a latent unit's localization in physical space, we fit each
latent unit distribution, with respect to physical space, to a two-dimensional Gaussian distribution
(\reffig{fig:4}{(\textbf{c}), top})
\[
  p = \frac{1}{2 \pi |\Sigma|} \exp \left[ -\frac{1}{2} (x - \mu)^{\intercal} \Sigma^{-1} (x - \mu)
  \right]
\]
We measure the area of the ellipsoid given by the Gaussian approximation where $p \geq 0.0005$
(\reffig{fig:4}{(\textbf{c})}, bottom). The area of the latent unit approximation measures how
localized a unit is compared to the environment's area, which measures $40 \times 65 = 2,600$
lattice units. The latent unit approximations have a mean area of $254.6$ lattice units and a $80\%$
of areas are $< 352.6$ lattice units, which cover $9.79\%$ and $13.6\%$ of the environment,
respectively.

The units in the neural network's latent space provide a unique
combinatorial code for each spatial position. The aggregate of latent units covers the environment's entire
physical space. At each lattice block in the environment, we calculate the number of active latent
units (\reffig{fig:4}{(\textbf{d})}, left). The number of active latent units is different in
$87.6\%$ of the lattice blocks. Every lattice block has at least one active latent unit, which
indicates the aggregate of the latent units cover the environment's
physical space. { Moreover, to ensure the regions
remain stable across shifting landmarks, environment's
trees were removed and randomly redistributed in the environment
(\reffig{fig:supp-E}{(\textbf{a-b})}). The regions remain stable after
changing the tree landmarks with the Jaccard index ($|S_\text{new}
\cap S_\text{old} | / | S_\text{new} \cup S_\text{old} |$) (or the
intersection over union between new regions $S_\text{new}$ and old $S_\text{old}$ regions) of 0.828.}

Lastly, we demonstrate that the neural network can measure physical distances and could perform
vector navigation---representing the vector heading from a current
location to a goal location---by comparing the combinations of overlapping regions
in its latent space. { We first
determine the active latent units by thresholding each continuous value by its 90th-percentile
value. At each position, we have a 128-dimensional binary vector that gives the overlap of 128
latent units. We take the bitwise difference $z_1 - z_2$ between the overlapping
codes $z_1$ and $z_2$ at two varying positions $x_1$ and $x_2$ with
the vector displacement $x_1 - x_2$
(\reffig{fig:supp-A}{(\textbf{a})}). We then fit a linear decoder from
the code $z_1 - z_2$ to the vector displacement $x_1 - x_2$,
\[
    x_1 - x_2 = W[z_1 - z_2] + b.
  \]
The predicted distance error $\norm{r - \hat{r}}$ and the predicted
direction error $\norm{\theta - \hat{\theta}}$ are decomposed from the
predicted displacement $\hat{x_1} - \hat{x_2}$. The linear decoder has
a low prediction error for distance ($< 80\%, 12.49$ lattice units;
mean $7.89$ lattice units) and direction ($< 80\%, 48.04\degree$; mean
$30.6\degree$) (\reffig{fig:supp-A}{(\textbf{b-c})}). The code $z_1 - z_2$
is highly correlated with direction $\theta$ and distance $r$ with
Pearson correlation coefficients $0.924$ and $0.718$, respectively
(\reffig{fig:supp-A}{(\textbf{d})}).

We can measure the correspondence between the bitwise
distance\footnote{In previous sections, we compute the Euclidean
distance between latent units. For the bitwise distance, we threshold
the latent units to its 90th-percentile then compute the $L_1$-norm between the units.}
$| z_1 - z_2 |$ and the physical distances $\norm{x_1 -
x_2}_{\ell_2}$. Similar the previous sections, we compute the joint densities of the binary vectors'
bitwise distances and the physical positions' Euclidean distances. We then calculate their mutual
information to measure how much spatial information the bitwise distance captures. The proposed
mechanism for the neural network's distance measurement---the binary vector's Hamming
distance---gives a mutual information of $0.542$ bits, compared to the predictive coder's mutual
information of $0.627$ bits and the auto-encoder's mutual information of $0.227$ bits
(\reffig{fig:4}{(\textbf{e})}). The code from the overlapping regions capture a
majority amount of the predictive coder's spatial information.
}

  \section*{Discussion}

Mapping is a general mechanism for generating an internal
representation of sensory information. While spatial maps facilitate
navigation and planning within an environment, mapping is a ubiquitous
neural function that extends to representations beyond visual-spatial
mapping. The primary sensory cortex (S1), for example, maps tactile
events topographically. Physical touches that occur in proximity are
mapped in proximity for both the neural representations and the
anatomical brain
regions\autocite{rosenthalS1RepresentsMultisensory2023}. Similarly,
the cortex maps natural speech by tiling regions with different words
and their relationships, which shows that topographic maps in the
brain extend to higher-order cognition. Similarly, the cortex maps
natural speech by tiling regions with different words and their
relationships, which shows that topographic maps in the brain extend
to higher-order cognition. The similar representation of non-spatial
and spatial maps in the brain suggests a common mechanism for charting
cognitive \autocite{behrensWhatCognitiveMap2018}.  However, it is
unclear how a single mechanism can generate both spatial and
non-spatial maps.

Here, we show that predictive coding provides a basic, general
mechanism for charting spatial maps by predicting sensory data from
past sensory experiences{---including environments
with degenerate observations}. Our theoretical framework applies to any
vector valued sensory data and could be extended to auditory data,
tactile data, or tokenized representations of language. We demonstrate
a neural network that performs predictive coding can construct an
implicit spatial map of an environment by assembling information from
local paths into a global frame within the neural network’s latent
space. The implicit spatial map depends specifically on the sequential
task of predicting future visual images . Neural networks trained as
auto-encoders do not reconstruct a faithful geometric representation
in the presence of physically distant yet visually similar landmarks.

Moreover, we study the predictive coding neural network’s representation in latent space. Each unit in the network’s latent space activates at distinct, localized regions—called place fields—with respect to physical space. At each physical location, there exists a unique combination of overlapping place fields. At two locations, the differences in the combinations of overlapping place fields provides the distance between the two physical locations. The existence of place fields in both the neural network and the hippocampus\autocite{okeefePlaceUnitsHippocampus1976} suggest that predictive coding is a universal mechanism for mapping. In addition, vector navigation emerges naturally from predictive coding by computing distances from overlapping place field units. Predictive coding may provide a model for understanding how place cells emerge, change, and function. 

Predictive coding can be performed over any sensory modality that has some temporal sequence. As natural speech forms a cognitive map, predictive coding may underlie the geometry of human language. Intriguingly, large language models train on causal word prediction, a form of predictive coding, build internal maps that support generalized reasoning, answer questions, and mimic other forms of higher order 
reasoning\autocite{brownLanguageModelsAre2020}. Similarities in spatial and non-spatial maps in the brain suggest that large language models organize language into a cognitive map and chart concepts geometrically. These results all suggest that predictive coding might provide a unified theory for building representations of information—connecting disparate theories including place cell formation in the hippocampus, somatosensory maps in the cortex, and human language.

  \section*{Acknowledgements}

We deeply appreciate Inna Strazhnik for her exceptional contributions to the scientific visualizations and figure illustrations. Her expertise in translating our research into clear visuals has significantly elevated the clarity and impact of our paper. We express our heartfelt gratitude to Thanos Siapas, Evgueniy Lubenov, Dean Mobbs, and Matthew
Rosenberg for their invaluable and insightful discussions which profoundly enriched our work. Their
expertise and feedback have been instrumental in the development and realization of this
research. Additionally, we appreciate the insights provided by Lixiang Xu, Meng Wang, and Jieyu
Zheng, which played a crucial role in refining various aspects of our study. The dedication and
collaborative spirit of this collective group have truly elevated our research, and for that, we are
deeply thankful.

  \section*{Code Availability}
The code supporting the conclusions of this study is available on GitHub at \href{https://github.com/jgornet/predictive-coding-recovers-maps}{https://github.com/jgornet/predictive-coding-recovers-maps}. The repository contains the Malmo environment code, training scripts for both the predictive coding and autoencoding neural networks, as well as code for the analysis of predictive coding and autoencoding results. Should there be any questions or need for clarifications about the codebase, we encourage readers to raise an issue on the repository or reach out to the corresponding author.

\section*{Data Availability}
All datasets supporting the findings of this study, including the latent variables for the autoencoding and predictive coding neural networks, as well as the training and validation datasets, are available on GitHub at \href{https://github.com/jgornet/predictive-coding-recovers-maps}{https://github.com/jgornet/predictive-coding-recovers-maps}. Researchers and readers interested in accessing the data for replication, verification, or further studies can contact the corresponding author or refer to the supplementary materials section for more details.

\section*{Notes}
\renewcommand{\thefootnote}{\arabic{footnote}}
{\footnotesize \footnotemark[1] Regarding the increasing number of
channels in the middle, we initially ran experiments in the
auto-encoder and the predictive coder with both the bottleneck
architecture (decreasing middle channels) and the fan-out architecture
(increasing middle channels). Both the bottleneck and fan-out
architectures gave similar performance in visual prediction—for both
the predictive coder and auto-encoder. In the predictive coder, we
found that the latent variables in the fan-out architecture generate
place cell-like fields seen in Figure 4. Because the auto-encoder has
similar performance for the bottleneck and fan-out architectures, we
use the fan-out architecture for the auto-encoder to provide a
controlled comparison between auto-encoding and predictive coding.

\footnotemark[2] The appearance of place cell-like firing is common in simpler networks
that perform spatial navigation such as Treves Miglino, Parisi (2007)
and Sprekeler \& Wiskott (2007). It is currently unclear whether
artificial place cell-like behavior corresponds to biological place
cells. Artificial place cell-like behavior could be an
artifact of the simplified inputs or spatial coordinates.

}

  {\footnotesize \printbibliography}
\end{multicols}

\newpage
\section*{Supplementary Information}
\beginsupplement
\label{sec:si}
{%
\subsection*{Neural networks solve predictive coding by performing maximum
likelihood estimation}\label{sec:app-nn}

We can express the model distribution $p_\theta(o_t | o_{<t})$ as
\begin{align*}
  p_\theta(o_t | o_{<t}) &= \int p_\theta(o_t, x_t, x_{<t} | o_{<t}) \; dx_t dx_{<t} \\
                         &= \int p_\theta(o_t | x_t) p_\theta(x_t | x_{<t}) p_\theta(x_{<t} |
                           o_{<t})\; dx_t dx_{<t}\; \\
                         &= \mathbb{E}_{x \sim p_\theta(x_t | o_{<t})} \left[p_\theta(o_t |
                           x_t)\right]
\end{align*}

Performing maximum likelihood estimation,
\begin{align*}
  &\argmax_{\theta} \mathbb{E}_{o \sim p_\text{data}(o)} p_\theta(o_1, \ldots, o_T)  \\
  &\quad= \argmax_{\theta} \sum^T_{t=1} \mathbb{E}_{o \sim p_\text{data}(o)} \left[ \mathbb{E}_{x_t \sim p_\theta(x_t | o_t)}
    p_\theta(o_t | x_t) \right] \\
\end{align*}
As $\log$ is a monotonic, increasing function, we can take perform maximum \textit{log-likelihood}
estimation,
\begin{align*}
  &\argmax_{\theta} \mathbb{E}_{o \sim p_\text{data}(o)} p_\theta(o_1, \ldots, o_T)  \\
  &\quad= \argmax_{\theta} \mathbb{E}_{o \sim p_\text{data}(o)} \log p_\theta(o_1, \ldots, o_T)  \\
  &\quad= \argmax_{\theta} \sum^T_{t=1} \mathbb{E}_{o \sim p_\text{data}(o)} \mathbb{E}_{x_t \sim p_\theta(x_t | o_t)}
    \left[ \log p_\theta(o_t | x_t) \right] \\
\end{align*}
Predictive coding, which is solving for $p_\theta(o_t | x_t)$ and $p_\theta(x_t |
o_{<t})$, is equivalent to estimating the data-generating distribution $p_\theta(o_1, \ldots, o_T)$.

Suppose that the agent's path and observations are deterministic. First, the agent's next position
given its past positions
\[
    x_t = F_\theta(x_{<t}).
\]
Second, the agent's \textit{sequence} of past positions given its \textit{sequence} of observations is given by
\[
    x_{<t} = f_\theta(o_{<t}).
\]
We can then parameterize the model distributions
\[
    p_\theta(x_t | x_{<t}) = \delta_{F_\theta(x_{<t})}(x_t)
\]
and
\[
    p_\theta(x_{<t} | o_{<t}) = \delta_{f_\theta(o_{<t})}(x_{<t})
\] with neural networks $o_{<t}
\xmapsto{f_\theta} x_{<t}$ and $x_{<t} \xmapsto{F_\theta} x_t$. If every positions have a single
observation, we can parameterize the distribution $p_\theta(o_t | x_t)$ as
\[
  p_\theta(o_t | x_t) = \mathcal{N}(o_t; g_\theta(x_t), \omega^2 \mathbf{I}),
\]
then maximum likelihood estimation becomes
\begin{align*}
  &\argmax_{\theta} \mathbb{E}_{o \sim p_\text{data}(o)} p_\theta(o_1, \ldots, o_T)  \\
  &\quad= \argmax_{\theta} \sum^T_{t=1} \mathbb{E}_{o \sim p_\text{data}(o)} \mathbb{E}_{x_t \sim p_\theta(x_t | o_t)} \left[
    \log p_\theta(o_t | x_t) \right] \\
  &\quad= \argmax_{\theta} \sum^T_{t=1} \mathbb{E}_{o \sim p_\text{data}(o)} \left[
    \log p_\theta(o_t | x_t = F_\theta\circ f_\theta(o_{<t})) \right] \\
  &\quad= \argmin_{\theta} \mathbb{E}_{o \sim p_\text{data}(o)} \sum^T_{t=1}
    \norm{o_t - g_\theta \circ F_\theta \circ f_\theta(o_{<t})}_{\ell_2}^2.
\end{align*}

In neural networks, the errors are propagated by gradient descent
\begin{align*}
  \theta \leftarrow \theta - \nabla_\theta \norm{o_t - g_\theta \circ F_\theta \circ f_\theta(o_{<t})}_{\ell_2}^2.
\end{align*}

\begin{figure}[p!]
  \centering
  \includegraphics[width=\textwidth]{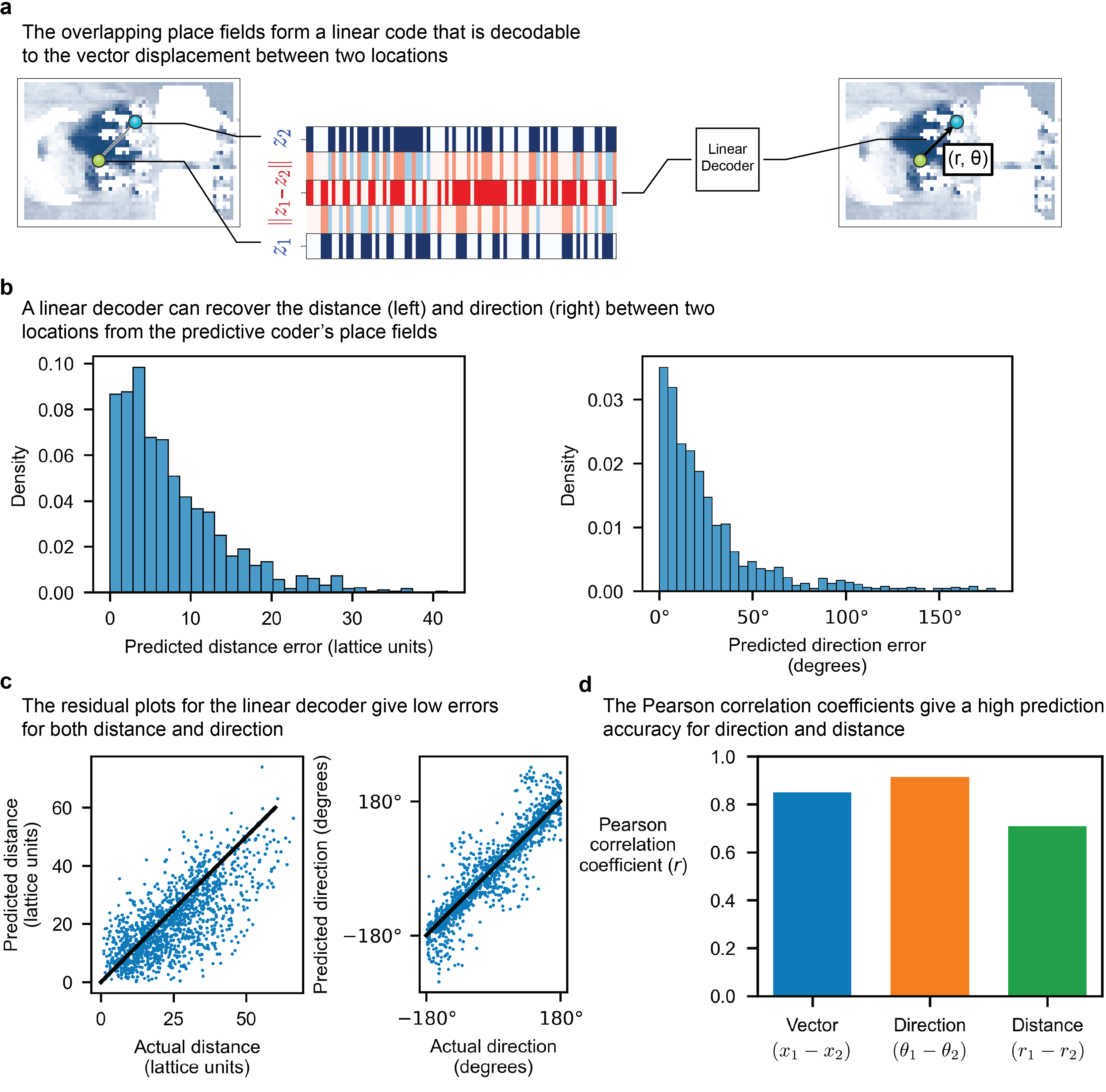}
  \caption{
\textbf{The place field overlap between two locations is linearly decodable to a
vector heading.}  { \textbf{a,} at different two locations $x_1$ and $x_2$, there
exists a place field code $z_1$ and $z_2$, respectively. The bitwise different
$z_1 - z_2$ gives the overlap between place fields at locations $x_1$ and
$x_2$. We perform linear regression inputting the overlap codes $z_1 - z_2$ and
predicting the vector displacement $x_1 - x_2$ between the two locations,
\begin{minipage}{\linewidth}
\[
    x_1 - x_2 = W [z_1 - z_2] + b.
\]
\end{minipage}
The linear model forms a linear decoder from the place field code $z_1 - z_2$ to
the vector displacement $x_1 - x_2$. \textbf{b,} the errors in predicted
distance $\norm{r - \hat{r}}$ (left) and predicted direction $\norm{\theta -
\hat{\theta}}$ (right) are decomposed from the predicted displacement $x_1 -
x_2$. The linear decoder has a low prediction error for distance ($< 80\%$,
12.49 lattice units; mean, $7.89$ lattice units) and direction ($< 80\%$,
48.04$\degree$; mean, $30.6\degree$). \textbf{c,} residual plots show both
distance (left) and direction (right) are strongly correlated with the place
field code. \textbf{d,} in addition, the place field code is strongly correlated
with the vector heading with Pearson correlation coefficients of $0.861$,
$0.718$, and $0.924$ for vector displacement ($x_1 - x_2$), distance ($r_1 - r_2$), and direction ($\theta_1 - \theta_2$), respectively.}
}\label{fig:supp-A}
\end{figure}

\subsection*{The predictive coding neural network requires
self-attention for an accurate environmental map}\label{sec:supp-1}
\begin{topfigure}[t!]
  \centering
  \includegraphics[width=\textwidth]{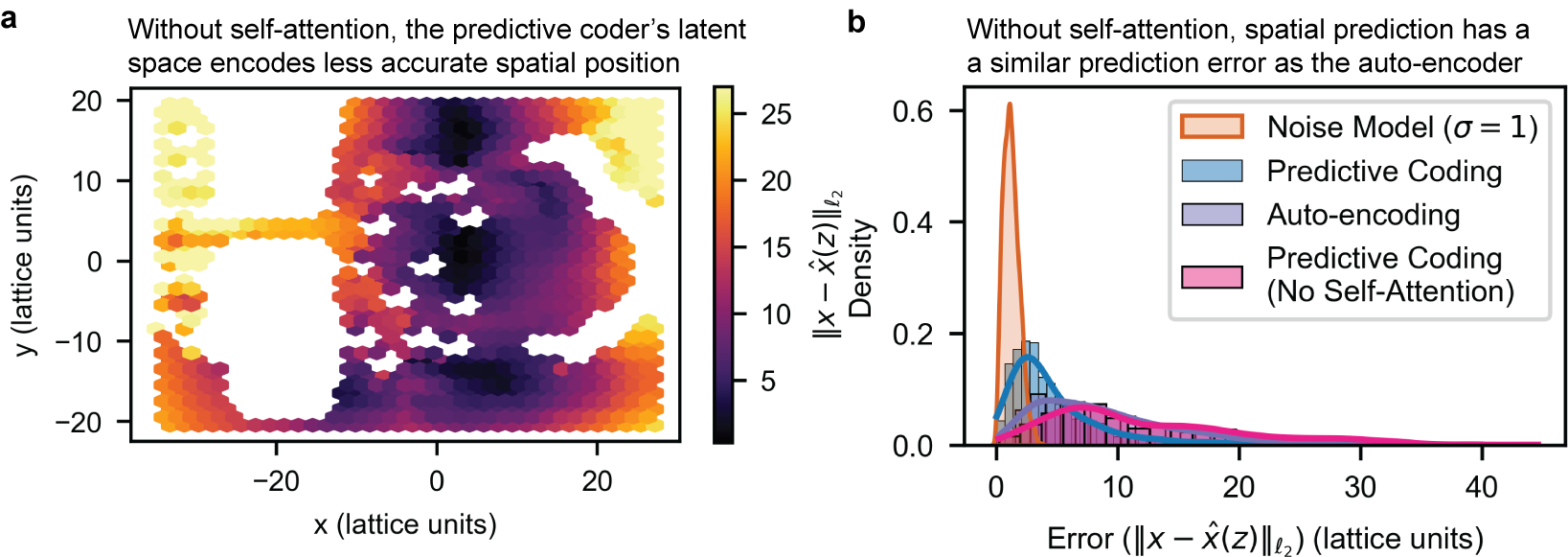}
  \caption{\textbf{Without self-attention, the predictive coder
encodes less accurate spatial information.} { \textbf{a-b,}
self-attention in the predictive coder captures sequential
information. To determine whether the temporal information is crucial
to build an accurate spatial map, a neural network predicts the
spatial location from the predictive coding's
latent space without self-attention. \textbf{a,} a heatmap of the
predictive coder's prediction error shows a low accuracy in many
regions. \textbf{b,} the histogram of prediction errors shows a
similar high prediction error as the auto-encoder.}}\label{fig:supp-1}
\end{topfigure}
The predictive coder architecture has three modules: encoder,
self-attention, and decoder. The encoder and decoder act on single
images---similar to how a tokenizer in language models transforms single
words to vectors. The self-attention operates on image sequences to
capture sequential patterns. If the predictive coder uses temporal
structure to build a spatial map, then the self-attention should build
the spatial map—not the encoder or decoder. In this section, we show
that the temporal structure is required to build an accurate map of
the environment.

Here we validate that self-attention is necessary to build an accurate
map. First, we take the latent units encoded by the predictive coder's
encoder (without the self-attention). We then train a separate neural
network to predict the actual spatial position given the latent
unit. The accuracy of the predicted positions provides a lower bound
on the spatial information given by the predictive coder's encoder’s
latent space. The heatmap (left) visualizes the errors given different
positions in the environment. The histogram of the prediction errors
(right) provides a comparison between the latent spaces of the
predictive coder, the predictive coder without self–attention, and the
auto-encoder. Without self-attention, the predictive coder has a much
higher prediction error of its position similar to the auto-encoder.

\newpage

\subsection*{The continuity and number of past observations determines
the environmental map's accuracy}

\begin{topfigure}[t!]  \centering \includegraphics[width=\textwidth]{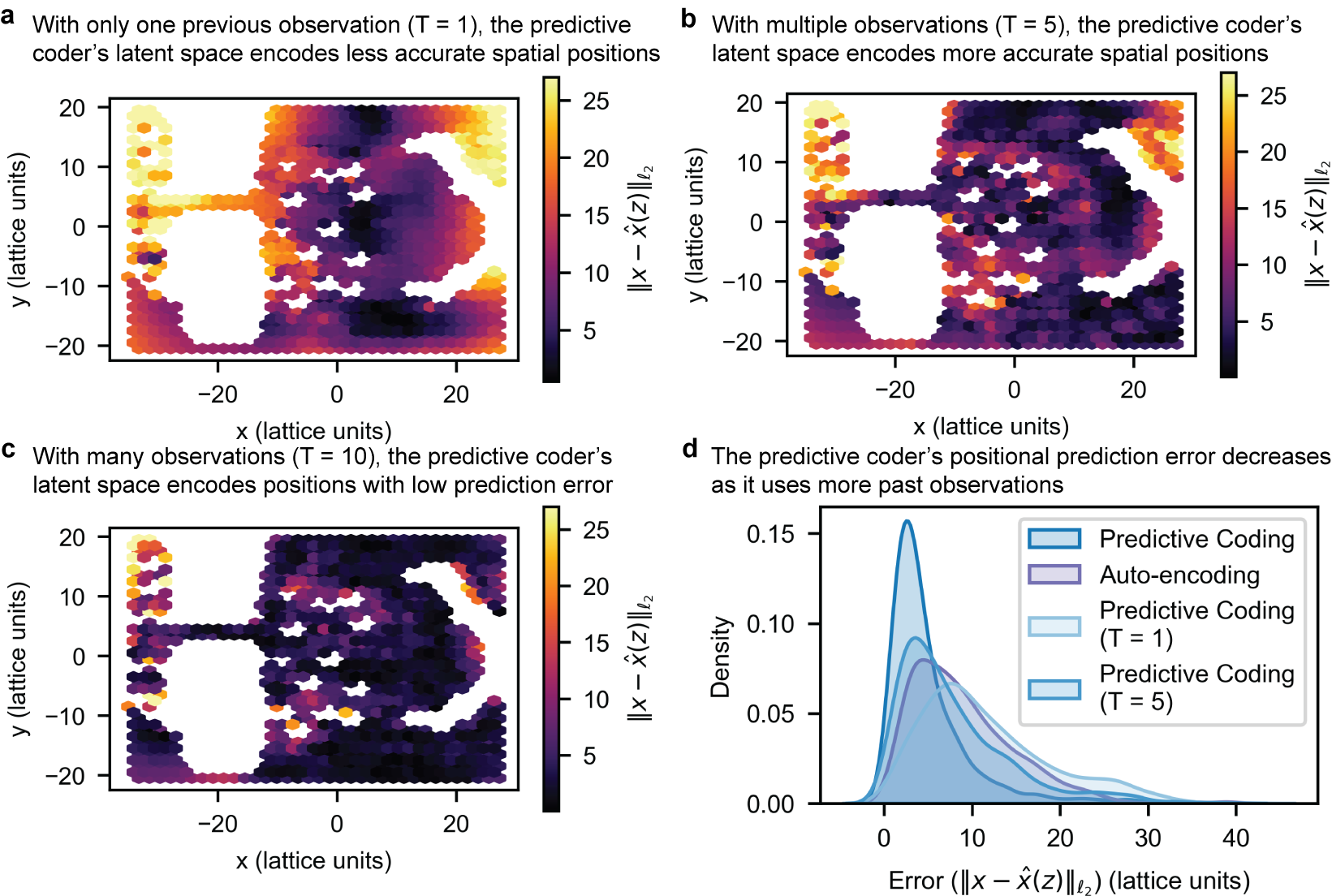} \caption{\textbf{As the number of past observations increase, the predictive coder's
positional prediction error decreases.} { \textbf{a-c,} the predictive coder trains with one, five,
and ten past observations, respectively. To determine how much temporal information is crucial to
build an accurate spatial map, a neural network predicts the spatial location from the predictive
coding's latent space. \textbf{a,} with only one past observation, a heatmap of the predictive
coder's prediction error shows a high error in many regions. \textbf{b,} with five past
observations, the prediction error reduces in many regions. \textbf{c,} with ten past observations,
the prediction error is reduced below 7.3 lattice units for the majority (> 80\%) of
positions. \textbf{d,} as the number of past observations goes to zero, the histogram of prediction
errors converges to the auto-encoder's prediction error.}}\label{fig:supp-2}
\end{topfigure}
In the predictive coder's architecture, it contains three modules: the encoder, the self-attention,
and the decoder. As discussed in the \hyperlink{sec:supp-1}{previous section}, the predictive
 coder's requires self-attention learn an accurate spatial map: the observation's temporal
 information is crucial to build an environment's map. A question that arises is how much temporal
 information does the predictive coder require to build an accurate map? In this section, we show that the predictive coder's spatial prediction error decreases as the continuity and number of past observations.

 First, we take the latent units encoded by the predictive coder's
 encoder trained on differing numbers of past observations
 ($T = 1, 5, 10$).
 We then train a separate neural network to predict the actual spatial
 position given the latent unit. The accuracy of the predicted
 positions provides a lower bound on the spatial information given by
 the predictive coder's encoder’s latent space. The heatmap
 (\reffig{fig:supp-2}{(\textbf{a}, \textbf{b}, \textbf{c})})
 visualizes the errors given different positions in the
 environment. With only one past observation, a heatmap of the
 predictive coder's prediction error shows a high error in many
 regions. With five past observations, the prediction error reduces in
 many regions. With ten past observations, the prediction error is
 reduced below 7.3 lattice units for the majority (> 80\%) of
positions. As the number of past observations goes to zero, the
histogram of prediction errors converges to the auto-encoder's
prediction error (\reffig{fig:supp-2}{(\textbf{d})}).

\begin{figure}[p!]  
  \centering
  \includegraphics[width=\textwidth]{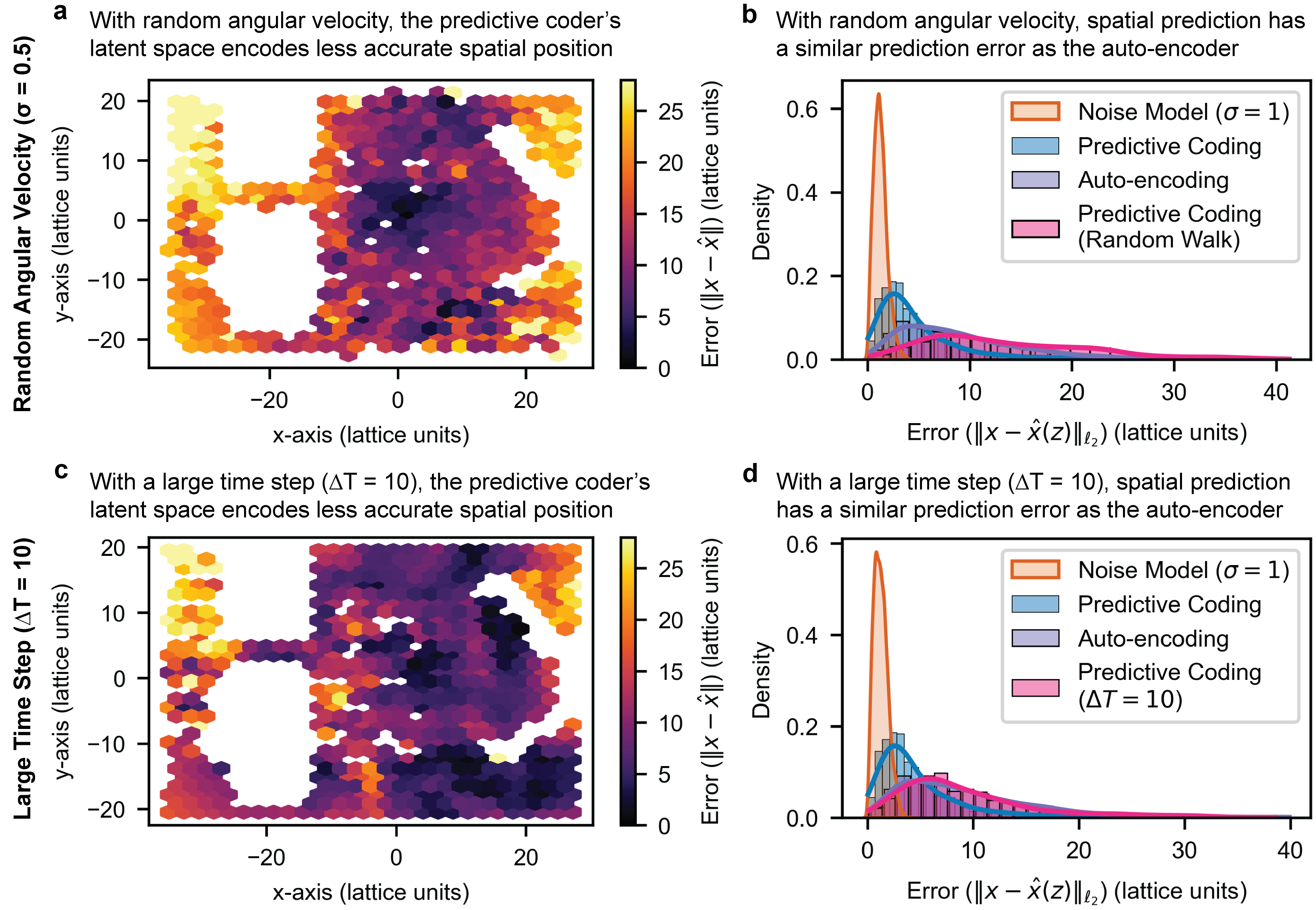} 
\caption{\textbf{As the number of past observations increase, the
predictive coder's positional prediction error decreases (cont.).}
{%
\textbf{a-b,} the agent typically traverses the environment in direct
path with minimal angular rotation. To determine the impact of angular
rotation on the predictive coder, the agent samples a random angular
velocity ($\sigma = \nicefrac{30\degree}{\text{sec}}$) as it traverses
the environment. The positional prediction error (\textbf{a})
increases and the error density \textbf{(b)} shifts to the
auto-encoder's error density. \textbf{c-d,} the agent typically takes
short time steps per an observation (20 images per second). To
determine the impact of the time step length, the agent samples the
environment's images with a lower frame rate (2 images per
second). Similar to the random angular velocity, the large time step
results in the predictive coder having a higher prediction error
\textbf{(c)} and an error density \textbf{(d)} shifting toward the
auto-encoder's error density.}}\label{fig:supp-D}
\end{figure}

\begin{figure}[p!]  
  \centering
  \includegraphics[width=\textwidth]{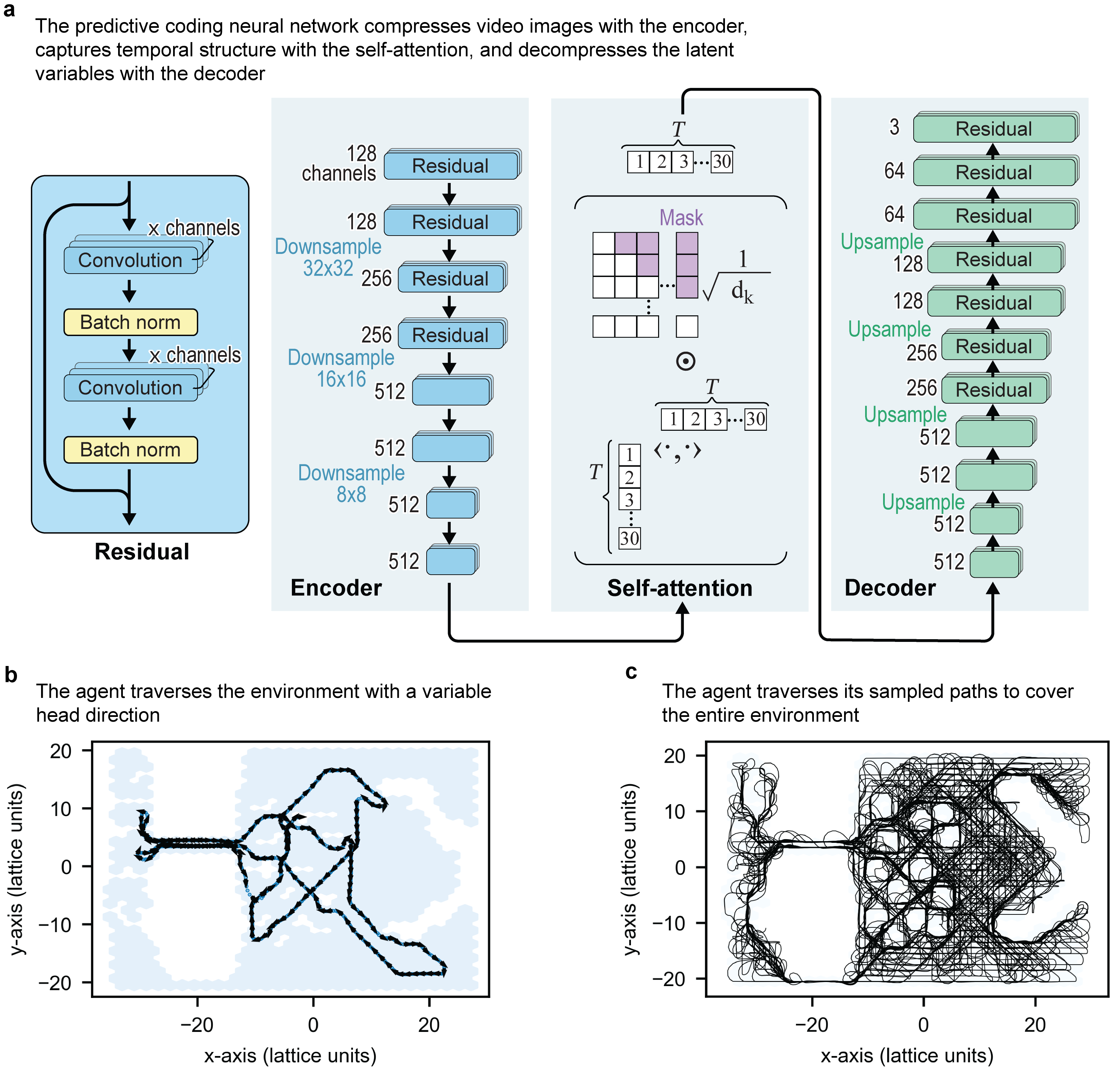} 
\caption{\textbf{Extended neural network architecture and training
description.} { \textbf{a,} the predictive coding neural network, or
predictive coder, uses an encoder, self-attention, and decoder to
perform predictive coding. The encoder is a convolutional neural
network architecture called ResNet-18 that uses residuals to compress
the high-dimensional video image. The self-attention module capture
temporal dependencies from the low-dimensional encoded images. The
self-attention's output gives the predictive coder's latent
variables. The decoder is a convolutional neural network that
upscales---rather than downscales---the predictive coder's latent
variables to predicted images. \textbf{b,} an example path of the
agent shows the agent traversing the environment with a variable head
direction. \textbf{c,} the agent's paths traverse the space to cover the entire
environment.}}\label{fig:supp-E}
\end{figure}

\begin{figure}[p!]  
  \centering
  \includegraphics[width=\textwidth]{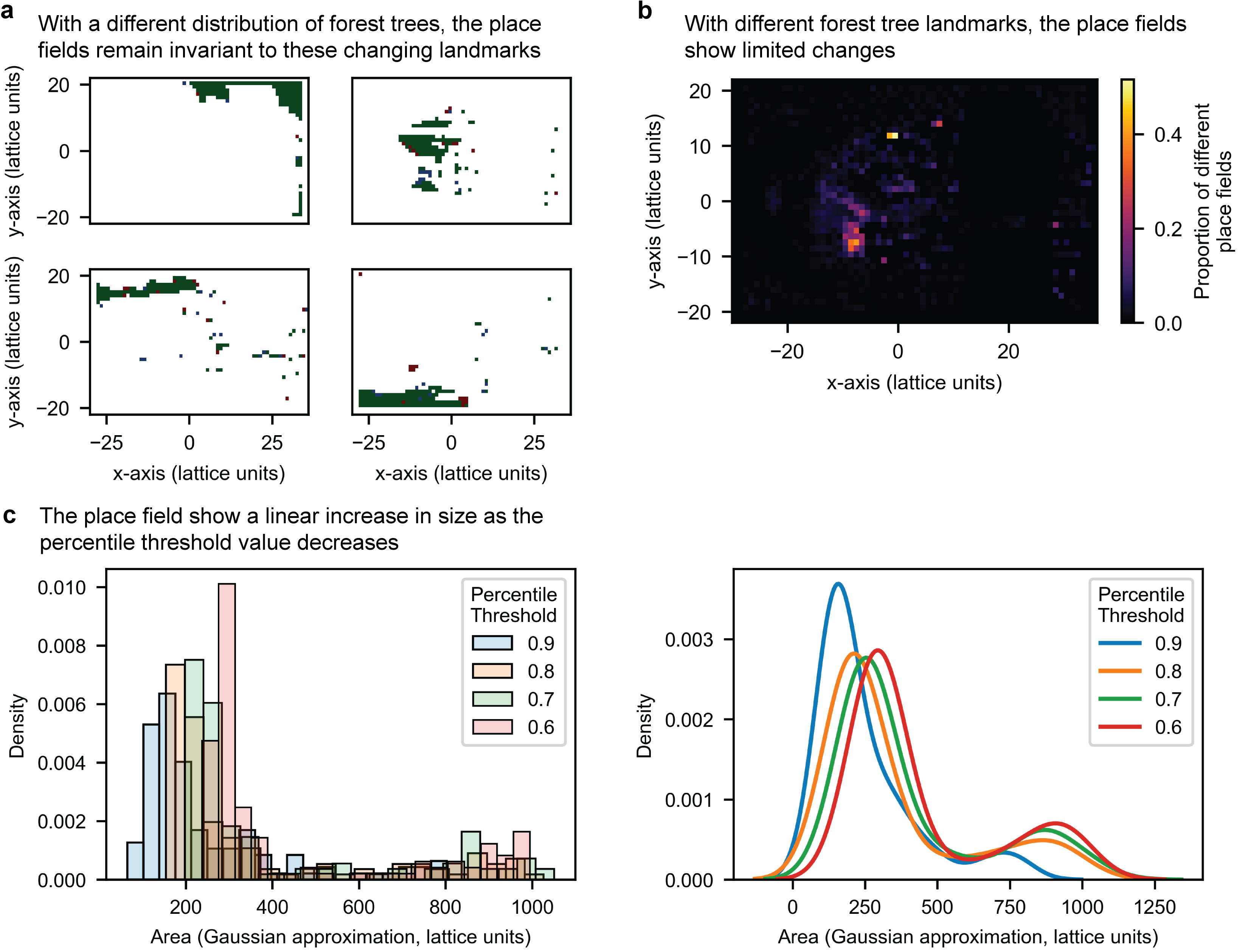} 
\caption{\textbf{Extended place field results.} { \textbf{a-b,} The 
predictive coder's place field latent variables are invariant to
shifting landmarks. To determine the effect of shifting landmarks,
the trees in the environment were removed and randomly redistributed
in the forest region.  \textbf{a,} the predictive coder's original
place fields (\textbf{red}) were overlaid with the new place fields
(\textbf{blue}), and the union set of the original and new place
fields are shown in green. The new, shifted landmark place fields
demonstrate a large overlap (Jaccard index ($| A \cap B |/ | A \cup B
|$) = 0.828) with the original
place fields. \textbf{b,} a visual overlap of the proportion of
different place fields at every location. The place fields show no
variability outside the forest and low variability inside the forest
region. \textbf{c,} the place fields are measured by thresholding the
predictive coder's latent variable, so the place field sizes are dependent on the
percentile threshold value. To determine the effect of the threshold
value on the place field sizes, the place field size histogram (left)
is plotted with respect to the percentile threshold value, and the
place field size densities (right) are estimated using kernel density
estimation from the histogram.}}\label{fig:supp-F}
\end{figure}

}

\end{document}